\documentclass[12pt, a4paper]{article}
\usepackage[utf8]{inputenc}
\usepackage[T1]{fontenc}
\usepackage[english]{babel}
\usepackage{textcomp}
\usepackage{amsmath, amssymb, amsthm, mathrsfs}
\usepackage{stmaryrd}
\usepackage[unicode, bookmarksnumbered]{hyperref}
\usepackage{bookmark}
\hypersetup{bookmarksopen=true, bookmarksopenlevel=2}

\newcommand{\FinTrace}{\mathsf{FinTrace}}
\newcommand{\InfTrace}{\mathsf{InfTrace}}
\newcommand{\Ty}{\mathsf{Ty}}
\newcommand{\Prop}{\mathsf{Prop}}
\newcommand{\restrict}{\mathsf{restrict}}
\newcommand{\Type}{\mathsf{Type}}
\newcommand{\TypeL}{\mathsf{Type}_\ell}
\newcommand{\TypeLi}[1]{\mathsf{Type}_{\ell,#1}}
\newcommand{\Sub}{\mathsf{Sub}}
\newcommand{\Ext}{\mathsf{Ext}}
\newcommand{\Extf}{\mathsf{Ext}_f}
\newcommand{\Tm}{\mathsf{Tm}}
\newcommand{\Uc}{\mathsf{Uc}}
\newcommand{\State}{\mathsf{State}}
\newcommand{\Event}{\mathsf{Event}}
\newcommand{\Nat}{\mathsf{Nat}}
\newcommand{\nil}{\mathsf{nil}}
\newcommand{\step}{\mathsf{step}}
\newcommand{\head}{\mathsf{head}}
\newcommand{\tail}{\mathsf{tail}}
\newcommand{\refl}{\mathsf{refl}}
\newcommand{\prefix}{\mathsf{prefix}}
\newcommand{\Paths}{\mathsf{Paths}}
\newcommand{\Step}{\mathsf{Step}}
\newcommand{\Ux}{\mathsf{U}}
\newcommand{\Id}{\mathsf{Id}}
\newcommand{\Bool}{\mathsf{Bool}}
\newcommand{\subsubsubsection}[1]{\paragraph{#1}\mbox{}\par\noindent\ignorespaces}

\newtheorem{theorem}{Theorem}[section]
\newtheorem{lemma}[theorem]{Lemma}
\newtheorem{corollary}[theorem]{Corollary}
\newtheorem{definition}[theorem]{Definition}

\title{NM-DEKL$^3_\infty$: A Three-Layer Non-Monotone Evolving Dependent Type Logic}
\author{Peng Chen}
\date{\today}

\begin{document}

\maketitle

\begin{abstract}
We present a new dependent type system, NM-DEKL$^3_\infty$ (Non-Monotone Dependent Knowledge-Enhanced Logic), for formalising evolving knowledge in dynamic environments. The system uses a three-layer architecture separating a computational layer, a constructive knowledge layer, and a propositional knowledge layer. We define its syntax and semantics and establish Soundness and Equational Completeness; we construct a syntactic model and prove that it is initial in the category of models, from which equational completeness follows. We also give an embedding into the $\mu$-calculus and a strict expressiveness inclusion (including the expressibility of non-bisimulation-invariant properties).
\end{abstract}

\clearpage
\tableofcontents
\clearpage

\section{Introduction}
NM-DEKL$^3_\infty$ models evolving knowledge in dynamic environments around \textbf{three layers} ($\Uc$/$\TypeL$/$\Prop$), \textbf{dual traces} (finite/infinite), \textbf{non-monotonicity} (presheaf contravariance), and a \textbf{computational kernel} ($\Uc$ independent of $\TypeL$/$\Prop$). Derivation is construction: $\Gamma\vdash k:K_f(\tau)$ means there is a knowledge witness $k$; non-monotonicity is captured by the functor laws of restriction and presheaf contravariance (the definition of failure and its relation to non-monotonicity are in Definition~\ref{def:failure}, Theorem~\ref{thm:nonmono-iff-failure}, and Section~\ref{sec:semantics-nonmono}). The paper gives syntax and categorical semantics, Soundness/equational completeness/initiality, and embedding and expressiveness results relative to the $\mu$-calculus; the system is built on dependent type theory, categorical semantics (CwF, presheaf), and the modal $\mu$-calculus\cite{martinlof84,dybjer96,awodey09,kozen83}.

In the landscape of type theory, NM-DEKL$^3_\infty$ differs from static MLTT/HoTT: it combines dependent types with \emph{knowledge evolution} and \emph{non-monotonic reasoning}, formalising the dynamic process in which knowledge can be restricted or can fail as the trace extends. Applications include formalising dynamic knowledge reasoning, causal and counterfactual reasoning, and modelling under uncertainty or knowledge update (e.g. medical or legal reasoning).

\subsection{Relation to Existing Systems}
For dependent type theory and categorical semantics we rely on Martin-L\"of type theory\cite{martinlof84}, CwF and internal type theory semantics\cite{dybjer96,pitts00}, category theory\cite{awodey09}, and type-theoretic semantics\cite{streicher91}; for propositional fixed points and modality we refer to the $\mu$-calculus\cite{kozen83,bradfield07}. Compared with \textbf{MLTT/HoTT}: the latter focus on static types and (homotopical) equality, with models as static categories or $\infty$-groupoids\cite{hott13}, whereas NM-DEKL$^3_\infty$ uses presheaf and $\restrict$ to model a dynamic, non-monotonic constructive layer. Compared with the \textbf{$\mu$-calculus}: we give an embedding of the Prop layer into the $\mu$-calculus and a strict expressiveness inclusion (including non-bisimulation-invariant properties); see the embedding and expressiveness sections. Compared with \textbf{LTL/CTL}: temporal and path quantification embed into the Prop layer with semantics preserved\cite{pnueli77,emerson81,clarke86}; see the LTL/CTL embedding theorem.

\section{Syntax of NM-DEKL$^3_\infty$}

\paragraph{Notational Conventions}
We use the following symbols consistently throughout (avoid mixing with other notations):
\begin{itemize}
  \item \textbf{Constructive layer universe} $\TypeL$ (with levels $\TypeLi{i}$); we do not use Type$_\ell$, Type\_L, etc.
  \item \textbf{Constructive knowledge fibre} $K_f$: syntactically a type family $\FinTrace\to\TypeL$, semantically a presheaf $K_f:\mathcal{T}_f^{\mathrm{op}}\to\mathcal{C}$.
  \item \textbf{Trajectory extension type} $\Extf(\tau,\tau')$; \textbf{restriction term constructor} $\restrict(\epsilon,k)$: restricts a term $k$ of $K_f(\tau')$ along extension $\epsilon:\Extf(\tau,\tau')$ to a term of $K_f(\tau)$. $\Extf$ and $\restrict$ together express non-monotonicity (see Section~\ref{sec:semantics-nonmono}).
\end{itemize}

\medskip
\noindent\textbf{Symbol table (frequently used)}:
$\Uc$/$\Uc_i$ computational universes;
$\TypeL$/$\TypeLi{i}$ constructive universes;
$\Prop$ propositional layer;
$\FinTrace$/$\InfTrace$ finite/infinite trace;
$\State$, $\Event$, $\Nat$ state, event, natural numbers;
$\Step$ state transition;
$K_f$, $K_\infty$ knowledge fibres;
$\Extf$, $\restrict$ trace extension and restriction;
$\llbracket -\rrbracket$ semantic interpretation;
$\equiv$ definitional equality.
\medskip

Formal definitions follow; stratification is in Section~\ref{sec:stratification}.

\subsection{Three Universes and Stratification}

The system has three layers:
\[
\Uc_0 : \Uc_1 : \Uc_2 : \cdots
\quad \text{and} \quad
\TypeLi{0} : \TypeLi{1} : \TypeLi{2} : \cdots
\quad \text{and} \quad
\Prop : \TypeLi{0}
\]
The computational layer $\Uc$ does not depend on $\TypeL$ or $\Prop$, ensuring an independent computational kernel (see \emph{stratification} in Section~\ref{sec:stratification}).

\subsection{Computational Layer $\Uc$ and Traces}

Basic types of the computational layer:
\[
State : \Uc_0, \quad Event : \Uc_0, \quad Nat : \Uc_0
\]
denoting state, event, and natural numbers.

The state transition function $\mathsf{Step}$ is
\[
\Step : \State \to \Event \to \State \to \Uc_0
\]
(transition triggered by an event).

\paragraph{Basic Types}

\begin{align*}
\State : \Uc_0 \\
\Event : \Uc_0 \\
\Nat   : \Uc_0
\end{align*}

\subsubsection{Finite Traces (Inductive)}

Finite traces are given by
\[
\FinTrace : \Uc_0
\]
with constructors
\[
\begin{aligned}
nil  &: State \to \FinTrace \\
step &: \FinTrace \to Event \to State \to \FinTrace
\end{aligned}
\]
$\nil$ is the initial state; $\step$ extends the trace by one transition.

\subsubsection{Infinite Traces (Coinductive)}

Infinite traces are given by
\[
\InfTrace : \Uc_0
\]
with destructors
\[
\begin{aligned}
head &:\; \InfTrace \to State \\
tail &:\; \InfTrace \to Event \times \InfTrace
\end{aligned}
\]
$\head$ returns the current state; $\tail$ returns the next event and the remainder of the trace.

\subsection{Constructive Knowledge Layer $\TypeL$ and Non-Monotonicity}

The constructive layer is given by a presheaf category:
\[
\Ty_\ell(\Gamma) = [\mathcal{T}_f^{op}, \mathcal{C}/\Gamma]
\]
$\TypeL$ types are interpreted as type families indexed by the finite-trace category $\mathcal{T}_f$; each $\tau$ corresponds to a type $B(\tau)$.

\paragraph{Knowledge Families}

\begin{align*}
K_f : \FinTrace \to \TypeLi{0} \\
K_\infty : \InfTrace \to \TypeLi{0}
\end{align*}

\paragraph{Trajectory Extension}

\begin{align*}
\Extf : \FinTrace \to \FinTrace \to \Uc_0
\end{align*}

\paragraph{Non-Monotonic Restriction Rule}

\[
\frac{
\Gamma \vdash \epsilon : \Extf(\tau,\tau')
\quad
\Gamma \vdash k : K_f(\tau')
}{
\Gamma \vdash \restrict(\epsilon,k) : K_f(\tau)
}
\]

The system does not provide $K_f(\tau) \to K_f(\tau')$.

\subsection{Identity Types}

(Identity types are used to formalise trace and causal equivalence; see the semantics and completeness discussion.)

\[
\frac{
\Gamma \vdash A : \Ux_i \quad
\Gamma \vdash a:A \quad
\Gamma \vdash b:A
}{
\Gamma \vdash \Id_A(a,b) : \Ux_i
}
\]

\[
\Gamma \vdash \refl_a : \Id_A(a,a)
\]

\subsubsection*{Identity Elimination (J)}

Given a dependent type family (motive):

\[
\Gamma \vdash
C :
\Pi(x:A).\,
\Pi(y:A).\,
\Pi(p:\Id_A(x,y)).\,
\Ux_j
\]

and a proof term for the reflexive case:

\[
\Gamma \vdash
d :
\Pi(x:A).\,
C(x,x,\refl_x)
\]

the elimination rule is:

\[
\frac{
\begin{aligned}
&\Gamma \vdash A : \Ux_i \\
&\Gamma \vdash C :
\Pi(x:A).\Pi(y:A).\Pi(p:\Id_A(x,y)).\Ux_j \\
&\Gamma \vdash d :
\Pi(x:A).C(x,x,\refl_x) \\
&\Gamma \vdash a : A \\
&\Gamma \vdash b : A \\
&\Gamma \vdash p : \Id_A(a,b)
\end{aligned}
}{
\Gamma \vdash
J_{A,C}(d,a,b,p)
:
C(a,b,p)
}
\]

\subsubsection*{Computation Rule ($\beta$-rule)}

\[
\Gamma \vdash
J_{A,C}(d,a,a,\refl_a)
\;\equiv\;
d(a)
:
C(a,a,\refl_a)
\]

\subsection{Constructive Knowledge and Causal Reasoning}

The constructive layer is used not only for types but also for \emph{causal reasoning}: combining traces and causal chains to strengthen reasoning.

\begin{definition}[Causal Reasoning]
\label{def:causal-reasoning}
In NM-DEKL$^3_\infty$, causal reasoning is expressed via \emph{causal chains}: if state $s$ reaches $s'$ via event $e$ (i.e.\ $\Step(s,e,s')$), one can form a one-step cause; given a finite trace $\tau$, if there is a sequence $e_1,e_2,\dots,e_n$ such that the state evolves from $s_0$ to $s_n$ via $\Step$, one can derive a causal proof chain along $\tau$. Formally, a causal proof along $\tau$ is:
\[
\mathrm{CausalProof}(\tau)
\;\Longleftrightarrow\;
\exists e_1,\dots,e_n:\Event.\;
s_0 \xrightarrow{e_1} s_1 \xrightarrow{e_2} \cdots \xrightarrow{e_n} s_n,
\]
where each step $s_i \xrightarrow{e_{i+1}} s_{i+1}$ corresponds to $\Step(s_i,e_{i+1},s_{i+1})$.
\end{definition}

\begin{theorem}[Constructive Knowledge and Logical Reasoning]
\label{thm:constructive-logic}
In NM-DEKL$^3_\infty$, the constructive layer can represent not only types but also \emph{richer logical reasoning} via causal chains and trace structure, including dynamic, causal, and history-dependent reasoning.
\end{theorem}

\begin{proof}[Proof idea]
Causal chains (Definition~\ref{def:causal-reasoning}) together with $K_f$ and $\restrict$ on traces express knowledge dependency along traces and causal chains between events. This is realised by the constructive layer and the trace/$\Step$ structure in the derivation system, so NM-DEKL$^3_\infty$ can express more complex reasoning than purely truth-valued propositions; see Section~\ref{sec:mu-embedding} for the non-bisimulation-invariant expressibility theorem relative to the $\mu$-calculus.
\end{proof}

\subsection{Propositional Knowledge Layer $\Prop$}

The propositional layer is interpreted as the subobject classifier:
\[
\Prop(\Gamma) = \Sub(1_\Gamma)
\]
Terms of type $\Prop$ are proof-irrelevant propositions.

\begin{align*}
\top,\bot : \Prop \\
\wedge,\vee,\to : \Prop \to \Prop \to \Prop \\
\forall_{x:A} P(x) : \Prop,\quad \exists_{x:A} P(x) : \Prop
\quad\text{(where }\Gamma\vdash A:\Uc_i,\;\Gamma,x:A\vdash P:\Prop\text{)}
\end{align*}

\subsection{$\mu/\nu$ Fixed Points (Prop Layer)}

If $F:\Prop \to \Prop$ is monotone:
\[
\mu X.F(X) : \Prop
\qquad
\nu X.F(X) : \Prop
\]
with fold/unfold: $F(\mu X.F) \leftrightarrow \mu X.F$. The \emph{layer separation theorem} (Theorem~\ref{thm:layer-separation}) ensures that $\mu/\nu$ in Prop does not break non-monotonicity of $\TypeL$.

\subsection{Syntactic Equality (Definitional Equality)}

\begin{definition}[Definitional equality $\equiv$]\label{def:definitional-equality}
In NM-DEKL$^3_\infty$, the \emph{definitional equality} $t \equiv t' : A$ holds when $t$ and $t'$ are interconvertible in the derivation system by a sequence of reduction/expansion steps; we write $t \equiv t'$.

Definitional equality is a transitive equivalence. Its \emph{generators} include (see the ``Definitional equality rules'' subsection):
\begin{itemize}
  \item $\beta$-rule (\emph{required}):
    \[
    (\lambda x.\,t)\,b \equiv t[b/x].
    \]
  \item $\eta$-expansion (if adopted): $\lambda x.(F\,x) \equiv F$ for $F : A \to B$ with $x\notin\mathrm{FV}(F)$.
  \item Computation rules: e.g.\ the $\beta$-rule for $J$ on $\refl$ for Identity types, functor laws for $\restrict$, as part of \emph{judgmental equality}.
\end{itemize}
\end{definition}

\subsection{Non-Monotonicity and Restriction}

Non-monotonicity is formalised by the syntactic rules and functor laws of restriction; semantically it corresponds to presheaf contravariance (see Section~\ref{sec:semantics-nonmono}).
\[
\restrict : \Extf(\tau, \tau') \to K_f(\tau') \to K_f(\tau)
\]

\paragraph{Formal definition of ``failure''}
In syntax and models, \textbf{failure} is defined as constructive unavailability:
\begin{itemize}
  \item \textbf{Syntax}: ``Knowledge fails'' on trace $\tau$ means there is no $k$ with $\Gamma \vdash k : K_f(\tau)$, i.e.\ $K_f(\tau)$ has no inhabitant (no evidence) in that context. Failure here is failure of a specific type at a given trace and context.
  \item \textbf{Model}: In a model $\mathcal{M}$, the trajectory $\tau$ \emph{fails} at $\Gamma$ when the fibre of $K_f(\tau)$ over $\llbracket\Gamma\rrbracket$ has no inhabitants; i.e.\ $\Tm_\ell(\llbracket\Gamma\rrbracket,\llbracket K_f(\tau)\rrbracket)$ is empty (or the subobject is initial).
\end{itemize}
So ``failure = non-constructibility'' formally means \emph{no evidence}; $\restrict$ yields a term of $K_f(\tau)$ from $K_f(\tau')$ only when there is $\epsilon : \Extf(\tau, \tau')$. If $\tau'$ is unreachable or there is no extension proof, one cannot construct evidence for $K_f(\tau)$, matching empty fibres in the model.

\begin{definition}[Failure]
\label{def:failure}
In NM-DEKL$^3_\infty$, given type $A$ and context $\Gamma$, we say $A$ \textbf{fails} at $\Gamma$ in model $\mathcal{M}$ if $A$ has no inhabitants at $\Gamma$ in $\mathcal{M}$, i.e.\ $\Tm(\llbracket\Gamma\rrbracket_{\mathcal{M}},\llbracket A\rrbracket_{\mathcal{M}})$ is empty.
If no term $t$ with $t : A$ can be constructed at $\Gamma$, we say the derivation $\Gamma \vdash t : A$ \textbf{fails}. So failure is ``non-constructibility'', not ``deriving negation''.
\end{definition}

\begin{theorem}[Relation between non-monotonicity and failure]
\label{thm:nonmono-iff-failure}
In NM-DEKL$^3_\infty$, non-monotonicity is a global property of the system, while failure refers to a type being non-constructible in a derivation. If a type becomes non-constructible (fails), the system exhibits non-monotonicity there; conversely, non-monotonicity (only restricting from long to short traces, no reverse construction) can make a type have no inhabitants on short traces, i.e.\ failure.
\end{theorem}

\begin{proof}
\textbf{Intuition}: Non-monotonicity is a local property of the system: \textbf{one cannot derive from long trace to short trace}. Failure is a type being ``non-constructible'' at a given trace. The link: if a type cannot be constructed at a trace (failure), then non-monotonicity shows up as derivations only along shortening traces, not from long to short.

\textbf{Formal}: By Definition~\ref{def:failure}, if type $A$ has no constructible term at some context or trace (i.e.\ $\Tm(\llbracket\Gamma\rrbracket,\llbracket A\rrbracket)$ is empty in that model), then $A$ fails there. This reflects non-monotonicity: one cannot extend from long to short trace in reasoning, because the short trace may have no constructible terms. So non-monotonicity can cause failure on some traces; the two are related.
\end{proof}

\paragraph{Sources of non-monotonicity (brief)}
Non-monotonicity in the system comes from: \textbf{(i)} when the trace shortens, the fibre of a type on the short trace may have no section (\emph{disappearance of the type}); \textbf{(ii)} when extension $\epsilon : \Extf(\tau, \tau')$ is unreachable or unprovable, one cannot get a term of $K_f(\tau)$ from $K_f(\tau')$ via $\restrict$ (\emph{path/extension cut}). In the model both appear as ``only contravariant restriction, no covariant construction'', corresponding to irreversibility of information under event-driven or state-backtracking semantics; see Section~\ref{sec:semantics-nonmono} for failure and reasoning paths.

\subsection{Judgement Forms and Calculus Rules}

This subsection collects judgement forms, $\Uc$ rules, trace rules, $\TypeL$ rules (including $\restrict$ and functor laws), Prop rules (including $\mu/\nu$ fold/unfold), and definitional equality for reference and meta-theoretic proofs; the preceding subsections are conceptual.

\subsubsection{Judgement Forms}

The basic judgement forms are:

\[
\begin{array}{ll}
\vdash \Gamma\;\mathsf{ctx} & \text{context well-formed} \\
\Gamma \vdash A : \Uc_i & \text{computational type} \\
\Gamma \vdash B : \TypeLi{i} & \text{constructive type} \\
\Gamma \vdash P : \Prop & \text{propositional type} \\
\Gamma \vdash t : A & \text{term} \\
\Gamma \vdash t \equiv t' : A & \text{definitional equality}
\end{array}
\]

with the three layers:
\[
\Uc_i \subseteq \text{computational layer},\quad
\TypeLi{i} \subseteq \text{constructive layer},\quad
\Prop \subseteq \TypeLi{0}.
\]

\paragraph{Stratification}\label{sec:stratification}
\begin{itemize}
  \item \textbf{Allowed}: Formation of $\TypeL$ and $\Prop$ types/propositions may depend on variables in $\Uc$ (e.g.\ $\Gamma,x:A\vdash B:\TypeL$ with $A:\Uc$).
  \item \textbf{Forbidden}: Formation of $\Uc$ types must not depend on $\TypeL$ or $\Prop$ (no rule from ``$\Gamma\vdash P:\Prop$'' to ``$\Gamma\vdash A:\Uc$'').
  \item \textbf{Forbidden}: Elimination from $\Prop$ to $\Uc$ (no ``if $\Gamma\vdash p:P$ then $\Gamma\vdash t:\Uc$''), so the computational kernel stays pure.
\end{itemize}

\subsubsection{Substitution}

Substitution is the basis for the initiality proof and CwF structure. Below is the minimal setup.

\paragraph{Substitution judgement}
\[
\Delta \vdash \sigma : \Gamma
\quad\text{($\sigma$ is a substitution from $\Delta$ to $\Gamma$)}
\]
Types and terms under substitution: $A[\sigma]$, $t[\sigma]$ (when $\Gamma\vdash A:\Uc_i$ and $\Delta\vdash\sigma:\Gamma$, $\Delta\vdash A[\sigma]:\Uc_i$; same for terms).

\paragraph{Composition and identity}
\[
(\sigma \circ \delta) : \Delta \to \Gamma
\quad\text{when }\Delta\vdash\sigma:\Psi,\;\Psi\vdash\delta:\Gamma\text{;}
\quad
\mathrm{id}_\Gamma : \Gamma \to \Gamma.
\]

\paragraph{Substitution laws in definitional equality}
\[
A[\mathrm{id}] \equiv A,\quad
A[\sigma\circ\delta] \equiv (A[\sigma])[\delta];
\quad\text{same for terms.}
\]

\begin{lemma}[Substitution Lemma]\label{lem:substitution}
If $\Delta\vdash\sigma:\Gamma$, then derivability is preserved for types and terms: if $\Gamma\vdash A:\Uc_i$ then $\Delta\vdash A[\sigma]:\Uc_i$; if $\Gamma\vdash t:A$ then $\Delta\vdash t[\sigma]:A[\sigma]$. Same for $\TypeL$ and $\Prop$. Identity and composition satisfy $A[\mathrm{id}]\equiv A$, $A[\sigma\circ\delta]\equiv (A[\sigma])[\delta]$ (and for terms).
\end{lemma}

\subsubsection{Context Rules}

\paragraph{Empty context}

\[
\frac{}{ \vdash \cdot\;\mathsf{ctx} }
\]

The empty context is well-formed. $\mathsf{ctx}$ is the meta-level judgement label.

\paragraph{Extension}

\[
\frac{
  \vdash \Gamma\;\mathsf{ctx}
  \qquad
  \Gamma \vdash A : \Uc_i
}{
  \vdash \Gamma, x:A\;\mathsf{ctx}
}
\]

If $\Gamma$ is well-formed and $\Gamma \vdash A : \Uc_i$, then $\Gamma, x:A$ is well-formed.
\subsubsection{Computational Layer $\Uc$ Rules}

\paragraph{Universe formation}

\[
\frac{}{
  \Gamma \vdash \Uc_i : \Uc_{i+1}
}
\]

\paragraph{Variable}

\[
\frac{
  \vdash \Gamma, x:A, \Delta\;\mathsf{ctx}
}{
  \Gamma, x:A, \Delta \vdash x : A
}
\]

\paragraph{$\Pi$-types}

\[
\frac{
  \Gamma \vdash A : \Uc_i
  \qquad
  \Gamma,x:A \vdash B : \Uc_i
}{
  \Gamma \vdash \Pi x:A.\,B : \Uc_i
}
\]

\paragraph{$\lambda$-introduction}

\[
\frac{
  \Gamma,x:A \vdash t : B
}{
  \Gamma \vdash \lambda x.\,t : \Pi x:A.\,B
}
\]

\paragraph{Application}

\[
\frac{
  \Gamma \vdash f : \Pi x:A.\,B
  \qquad
  \Gamma \vdash a : A
}{
  \Gamma \vdash f\,a : B[a/x]
}
\]

\subsubsection{Trace Rules}

\paragraph{Finite trace}

\[
\frac{
  \Gamma \vdash s : \State
}{
  \Gamma \vdash \mathsf{nil}(s) : \FinTrace
}
\]

\[
\frac{
  \Gamma \vdash \tau : \FinTrace
  \qquad
  \Gamma \vdash e : \Event
  \qquad
  \Gamma \vdash s : \State
}{
  \Gamma \vdash \mathsf{step}(\tau,e,s) : \FinTrace
}
\]

\paragraph{Extension relation}

\[
\frac{
  \Gamma \vdash \tau : \FinTrace
}{
  \Gamma \vdash \mathrm{id}_\tau :
  \Extf(\tau,\tau)
}
\]

\[
\frac{
  \Gamma \vdash \epsilon_1 : \Extf(\tau,\tau')
  \qquad
  \Gamma \vdash \epsilon_2 : \Extf(\tau',\tau'')
}{
  \Gamma \vdash
  \epsilon_2\circ\epsilon_1 :
  \Extf(\tau,\tau'')
}
\]

\subsubsection{Constructive Layer $\TypeL$ Rules}

\paragraph{$\TypeL$ Universe}

\[
\frac{}{
  \Gamma \vdash \TypeLi{i} : \TypeLi{i+1}
}
\]

\paragraph{Trajectory-indexed family}

\[
\frac{
  \Gamma \vdash \tau : \FinTrace
}{
  \Gamma \vdash K_f(\tau) : \TypeLi{0}
}
\]

\paragraph{restriction}

\[
\frac{
  \Gamma \vdash \epsilon : \Extf(\tau,\tau')
  \qquad
  \Gamma \vdash k : K_f(\tau')
}{
  \Gamma \vdash
  \restrict(\epsilon,k)
  :
  K_f(\tau)
}
\]

\paragraph{Restriction functor laws (judgmental equality)}

The following are \emph{definitional equality} (judgmental equality) rules ensuring $\TypeL$ forms a presheaf; functoriality holds when the syntactic model is quotiented by $\equiv$.
\[
\restrict(\mathrm{id},k)
\equiv
k
,\qquad
\restrict(\epsilon_1,
  \restrict(\epsilon_2,k))
\equiv
\restrict(\epsilon_2\circ\epsilon_1,k).
\]
\textbf{Example}: $\restrict(\mathrm{id},k)\equiv k$ means restriction along the identity extension does not change the evidence; $\restrict(\epsilon_1,\restrict(\epsilon_2,k))\equiv \restrict(\epsilon_2\circ\epsilon_1,k)$ means two restrictions along $\tau\to\tau''\to\tau'$ equal one restriction along the composite extension, matching the presheaf functor law.

\subsubsection{Prop Layer Rules}

\paragraph{Prop universe}

\[
\frac{}{
  \Gamma \vdash \Prop : \TypeLi{0}
}
\]

\paragraph{$\land$ introduction}

\[
\frac{
  \Gamma \vdash P : \Prop
  \qquad
  \Gamma \vdash Q : \Prop
}{
  \Gamma \vdash P \land Q : \Prop
}
\]

\paragraph{$\mu$-formation (monotonicity required)}

\[
\frac{
  \Gamma, X:\Prop \vdash \varphi(X) : \Prop
  \qquad
  \text{($\varphi$ monotone/positive in $X$)}
}{
  \Gamma \vdash \mu X.\varphi(X) : \Prop
}
\]

\paragraph{$\mu$ fold/unfold (proof terms)}

Introduction and elimination are given by \emph{proof terms}, not the logical symbol $\Rightarrow$:
\[
\frac{}{\Gamma \vdash \mathsf{fold} : \varphi(\mu X.\varphi) \to \mu X.\varphi}
\qquad
\frac{}{\Gamma \vdash \mathsf{unfold} : \mu X.\varphi \to \varphi(\mu X.\varphi)}
\]
They are inverses under definitional equality (e.g.\ $\mathsf{fold}(\mathsf{unfold}(w)) \equiv w$). $\nu$ is defined similarly ($\mathsf{out} : \nu X.\psi \to \psi(\nu X.\psi)$, $\mathsf{in}$, etc.).

\subsubsection{Identity Types}

\[
\frac{
  \Gamma \vdash A : \Uc_i
  \qquad
  \Gamma \vdash a : A
  \qquad
  \Gamma \vdash b : A
}{
  \Gamma \vdash \Id_A(a,b) : \Uc_i
}
\]

\[
\frac{
  \Gamma \vdash a : A
}{
  \Gamma \vdash \refl_a : \Id_A(a,a)
}
\]

Elimination rule $J$ omitted.

\subsubsection{Definitional Equality Rules}

\paragraph{$\beta$}

\[
(\lambda x.\,t)\,a
\equiv
t[a/x]
\]

\paragraph{$\eta$}

\[
\lambda x.\,(f\,x)
\equiv
f
\]

\paragraph{Equality closure}

\[
\frac{
  \Gamma \vdash t \equiv u : A
  \qquad
  \Gamma \vdash u \equiv v : A
}{
  \Gamma \vdash t \equiv v : A
}
\]

\[
\frac{
  \Gamma \vdash t \equiv u : A
}{
  \Gamma \vdash u \equiv t : A
}
\]

\section{Metatheory (I): Reduction and Normalisation}

\subsection{Subject Reduction}

If $\Gamma\vdash t:A$ and $t\to t'$, then $\Gamma\vdash t':A$ (types are preserved under reduction). The reduction and cofix strategy below are needed for the normalisation statements.

\subsection{Productivity}

For coinductive objects (e.g.\ $\InfTrace$),
$\mathsf{cofix}$ may be used to build infinite objects,
subject to a \emph{guardedness} condition.

\paragraph{Guardedness condition}

Let $\InfTrace$ be the following coinductive type:

\[
\InfTrace \;\cong\; State \times (Event \times \InfTrace)
\]

We allow
\[
\mathsf{cofix}\ f.\,t
\]
to build infinite traces,
provided the \emph{guardedness} condition holds:

\medskip

\noindent
In $t$, every use of the recursive variable $f$
must occur after one layer of observable structure,
i.e.\ inside the outer constructor of
\[
State \times (Event \times \InfTrace).
\]

\medskip

Formally, a valid definition has the form
\[
\mathsf{cofix}\ f.\;
\langle s,\,(e,\,f')\rangle
\]
where $f'$ may contain recursive calls to $f$,
but such calls must appear in the second component of $(e,\,\cdot)$.

\medskip

For example:

\[
\mathsf{cofix}\ f.\;
\langle s,\,(e,\,f)\rangle
\]
is valid,

whereas

\[
\mathsf{cofix}\ f.\; f
\]
is not.

\medskip

This ensures that for any object built with $\mathsf{cofix}$,
each application of a destructor
\[
head : \InfTrace \to State,
\qquad
tail : \InfTrace \to Event \times \InfTrace
\]
produces one layer of observable structure immediately,
giving \emph{productivity} of the coinductive object.

\subsection{Reduction Relation and Cofix Strategy}

To state normalisation precisely we fix the reduction relation and cofix unfold strategy.

\paragraph{Reduction $t \to t'$}
Generated by (closed under context and $\equiv$):
\begin{itemize}
  \item \textbf{$\beta$}: $(\lambda x.\,t)\,v \to t[v/x]$;
  \item \textbf{Identity $J$-$\beta$}: compute by definition when eliminating on $\refl$;
  \item \textbf{$\restrict$-$\beta$}: $\restrict(\mathrm{id},k) \to k$, $\restrict(\epsilon_1,\restrict(\epsilon_2,k)) \to \restrict(\epsilon_2\circ\epsilon_1,k)$.
\end{itemize}

\paragraph{Cofix unfold strategy}
Coinductive \emph{unfold} is ``expand on observation'': expand one step only when the term is observed by destructors \texttt{head}/\texttt{tail}, avoiding divergence from arbitrary unfold. \emph{Guardedness} is enforced syntactically: recursive calls must occur in strictly positive positions of constructors (e.g.\ $\mathsf{cons}$, $\step$).

\subsection{Normalisation}

With the three-layer structure,
normalisation can be proved layer by layer.

\subsubsection{Strong normalisation for the computational layer}

\begin{theorem}[Strong normalisation for $\Uc$]
\label{thm:Uc-normalization}
If
\[
\Gamma \vdash t : A
\quad\text{and}\quad
\Gamma \vdash A : \Uc_i,
\]
then $t$ is strongly normalising in the computational layer $\Uc$.
\end{theorem}

\begin{proof}[Proof idea]
The computational layer has:

\begin{itemize}
  \item $\Uc$ as a pure MLTT subsystem;
  \item no dependency on $\TypeL$ or $\Prop$;
  \item no $\mu/\nu$;
  \item no general recursion;
  \item no non-structural cofix.
\end{itemize}

So $\Uc$ is a standard strongly normalising dependent type system.

One can use Girard--Tait reducibility candidates or logical relations to show that all well-typed terms are strongly normalising. Hence the theorem.
\end{proof}

\subsubsection{Relative normalisation for $\TypeL$}

\begin{theorem}[Relative normalisation for $\TypeL$]
\label{thm:TypeL-normalization}
If
\[
\Gamma \vdash t : B
\quad\text{and}\quad
\Gamma \vdash B : \TypeLi{i},
\]
then $t$ is weakly normalising provided $\mathsf{cofix}$ is not unfolded.
\end{theorem}

\begin{proof}[Proof idea]

$\TypeL$ has:

\begin{itemize}
  \item Non-monotonic restriction only as structural maps, no new reduction paths;
  \item no $\mu$ in $\TypeL$;
  \item cofix only at $\InfTrace$;
  \item all cofix guarded.
\end{itemize}

So: no divergent constructive evidence; restriction does not increase reduction depth; evidence terms do not generate infinitely many new redexes. Thus weak normalisation holds under finite observation.
\end{proof}

\paragraph{Remark}

This is ``relative normalisation'': if cofix is not unfolded, there is no infinite reduction chain.

\subsection{Consistency}

\subsubsection{Empty proposition has no proof}

\begin{theorem}[Consistency]
\label{thm:consistency}
If
\[
\Gamma \vdash t : \bot,
\]
then the system is inconsistent.
\end{theorem}

\begin{proof}[Proof idea]

To show consistency we prove there is no closed term
\[
\vdash t : \bot.
\]
Without $\mu$ extension:

\begin{itemize}
  \item Prop is interpreted as the internal logic of the subobject classifier $\Omega$;
  \item $\bot$ as the initial subobject;
  \item if the base category is consistent (has a non-initial object), $\bot$ has no global section.
\end{itemize}

So in any consistent model, $\llbracket \bot \rrbracket = \emptyset$. By soundness, if $\vdash t:\bot$ then $\llbracket t \rrbracket$ would be an element of that empty object, a contradiction. So the system is consistent.
\end{proof}

\paragraph{Set model (consistency concretely)}

Take the base category $\mathcal{C}=\mathbf{Set}$ and $\mathcal{T}_f$ the finite-trace category (free category on $\State,\Event$). Then $\mathrm{PSh}(\mathcal{T}_f^{\mathrm{op}},\mathbf{Set})$ gives the presheaf semantics for $\TypeL$; $\Prop$ is the subobject classifier $\Omega$, $\bot$ the initial subobject (empty set/empty section). By Soundness there is no closed $\vdash t:\bot$, so $\llbracket\bot\rrbracket=\emptyset$ in this model; consistency follows.

\paragraph{Important note}

If $\mu/\nu$ is extended to $\Prop$, monotonicity conditions are needed to preserve consistency. With only monotone fixed points, consistency is still given by standard $\mu$-calculus semantics.

\paragraph{Summary}
Subject Reduction, strong normalisation ($\Uc$), relative normalisation ($\TypeL$), and consistency can all be established by layer.

\section{Categorical Semantics of NM-DEKL$^3_\infty$}

This section gives the model structure only, without repeating motivation. The computational layer is interpreted in a CwF; the constructive layer in the presheaf category $[\mathcal{T}_f^{\mathrm{op}},\mathcal{C}/\Gamma]$; the propositional layer as the subobject classifier $\Sub(1_\Gamma)$. Non-monotonicity in the model is given by the natural transformation for restriction: $\llbracket \restrict(\epsilon, k) \rrbracket = \llbracket K_f \rrbracket(m_\epsilon)(\llbracket k \rrbracket)$.

\medskip

We give the precise categorical semantics. First we describe the category of traces, then the three-layer model, and finally non-monotonicity and limit semantics for infinite knowledge.

\subsection{Basic categorical structure}

\subsubsection{Trajectory category $\mathcal{T}_f$}

\begin{definition}[Trajectory category $\mathcal{T}_f$]
Let $\State$ be the set of states and $\Event$ the set of events. Finite traces form a free category:
\[
\mathcal{T}_f = \mathrm{FreeCat}(\State,\Event).
\]
Structure:
\begin{itemize}
  \item \textbf{Objects}: states $s \in \State$;
  \item \textbf{Morphisms}: from $s$ to $s'$, all finite event sequences from $s$ to $s'$;
  \item \textbf{Composition}: concatenation of sequences;
  \item \textbf{Identity}: the empty sequence.
\end{itemize}
\end{definition}

$\mathcal{T}_f$ encodes the composition of finite evolutions: objects are states $\mathrm{Ob}(\mathcal{T}_f)=\State$, morphisms are finite event sequences (finite traces), composition is concatenation, identity is the empty sequence $\mathrm{id}_s:s\to s$. Event sequences as morphisms capture only the compositional structure; whether an event can occur at a state and the transition function are extra semantics not in the free category.

\subsubsection{Infinite trace category}

The free category $\mathcal{T}_f$ describes \emph{finite} composable structure: morphisms are finite event sequences and composition is concatenation. So it encodes ``finite evolution''. Infinite traces are better modelled as ``never-ending streams'', typically by coalgebra (observe/unfold) rather than algebra (generate).

\begin{definition}[Coalgebra structure for infinite traces]
Let the functor
\[
F(X) = \State \times (\Event \times X)
\]
be defined on a suitable category $C$. If a final coalgebra $\nu F$ exists, set
\[
\InfTrace := \nu F
\]
with structure map
\[
\mathrm{out} : \InfTrace \to \State \times (\Event \times \InfTrace).
\]
This captures infinite evolution.
\end{definition}

In $\mathbf{Set}$, if $\State$ and $\Event$ are sets, the final coalgebra exists (by standard coalgebra theory).

\paragraph{Meaning of $F(X)=\State \times (\Event \times X)$}

An element of an $F$-structure is: current state + next event + remainder (again an $F$-structure). So for an ``infinite trace object'' we can observe: the current state $s \in \State$; the next event $e \in \Event$; and the rest as an infinite trace. The structure map $\mathrm{out} : \InfTrace \to \State \times (\Event \times \InfTrace)$ decomposes an infinite trace into head (state + event) and tail (remainder).

\paragraph{Final coalgebra}

An $F$-coalgebra is a pair $(X,\gamma)$ with $\gamma : X \to F(X)$. The final coalgebra $(\nu F,\mathrm{out})$ is such that for any $(X,\gamma)$ there is a unique coalgebra morphism $\llbracket - \rrbracket : X \to \nu F$ with
\[
\mathrm{out} \circ \llbracket - \rrbracket = F(\llbracket - \rrbracket) \circ \gamma.
\]
So $\nu F$ is the canonical container of ``all possible infinite behaviours''; any system with structure $X \to \State \times (\Event \times X)$ maps uniquely into $\nu F$, giving its behaviour semantics.

\paragraph{Concrete form in $\mathbf{Set}$}

In $\mathbf{Set}$, $\nu F$ is (isomorphic to) infinite sequences $(s_0,e_0,s_1,e_1,s_2,e_2,\dots)$. The map $\mathrm{out}$ is head/tail: $\mathrm{out}(\text{full trace}) = (s_0, (e_0, \text{rest}))$.

\subsection{Three-layer semantic structure}

We use a fibred locally cartesian closed category (LCCC):

\begin{itemize}
  \item \textbf{Computational layer}: $\Uc \mapsto C$ with $C$ an LCCC.
  \item \textbf{Constructive layer}: $\TypeL \mapsto \mathrm{PSh}(\mathcal{T}_f, C) = [\mathcal{T}_f^{op}, C]$.
  \item \textbf{Propositional layer}: $\Prop \mapsto \Omega$, where $\Omega$ is the subobject classifier in $C$.
\end{itemize}

So $\Uc \mapsto C$, $\TypeL \mapsto \mathrm{PSh}(\mathcal{T}_f, C)$, $\Prop \mapsto \Omega$. This yields a fibred LCCC with the constructive layer as presheaf fibre over the computational layer.

\subsubsection{LCCC and dependent types}

Let $C$ be a locally cartesian closed category (LCCC). This is a key structural assumption.

\begin{itemize}
  \item A cartesian closed category (CCC) has $A \times B$ and $A \to B$.
  \item An LCCC additionally requires that for every object $X$, the slice $C/X$ is a CCC.
\end{itemize}

This corresponds to dependent type theory: $\Sigma_{x:X} A(x)$ and $\Pi_{x:X} A(x)$. So $C$ is LCCC iff the semantics supports dependent types, context substitution, and dependent function formation.

\subsubsection{Computational layer: $\Uc \mapsto C$}

The computational layer is interpreted in an LCCC: $\Uc \mapsto C$. It carries data (e.g.\ $\Nat$, $\Bool$), programs and functions, and basic execution structure. $C$ may be $\mathbf{Set}$, a domain-theoretic category, or a category of execution objects with resources. This layer is essentially the semantics of a standard programming language.

\subsubsection{Constructive layer: $\TypeL \mapsto [\mathcal{T}_f^{op},C]$}

The constructive layer is the presheaf category $\TypeL \mapsto \mathrm{PSh}(\mathcal{T}_f,C) = [\mathcal{T}_f^{op},C]$.

\paragraph{Meaning of presheaves}

A presheaf $F:\mathcal{T}_f^{op}\to C$ assigns to each state $s$ an object $F(s)\in C$ and to each trace $t:s\to s'$ a reindexing map $F(t):F(s')\to F(s)$. By contravariance, this pulls back structure from ``later'' to ``earlier'' and reinterprets evidence and types along traces. $[\mathcal{T}_f^{op},C]$ describes types or knowledge objects that vary with state and trace; $F(s)$ is what is constructible or verifiable at state $s$. Reindexing ensures evidence and objects are transported and audited consistently as the system evolves.

\subsubsection{Propositional layer: $\Prop \mapsto \Omega$}

The propositional layer is the subobject classifier: $\Prop \mapsto \Omega$. In $\mathbf{Set}$, $\Omega = \{\mathsf{true},\mathsf{false}\}$; in a general topos, $\Omega$ is the truth-value object. For every mono $m:U\hookrightarrow X$ there is a unique characteristic map $\chi_m:X\to\Omega$. So propositions = subobjects, properties = characteristic maps; proof theory is explained via subobjects and characteristic maps.

\subsubsection{Fibration structure}

The three layers are organised as a fibration:

\[
p:\TypeL\to\Uc.
\]

Meaning: when the base (computational layer) changes, the types above are reindexed. In dependent type semantics, substitution $\sigma$ corresponds to pullback. Since $C$ is an LCCC, slices support $\Pi/\Sigma$, so dependent types are preserved under reindexing. The three-layer structure is a fibred LCCC with the constructive layer as presheaf fibre over the computational layer. Summary: computational layer $C$ (base program semantics); constructive layer $[\mathcal{T}_f^{op},C]$ (state- and trace-indexed types/knowledge); propositional layer $\Omega$ (subobject classifier).

\subsection{Semantic theorem for non-monotonicity}\label{sec:semantics-nonmono}

\begin{theorem}[Restriction from contravariant structure]
Let $K_f : \mathcal{T}_f^{op} \to C$ be a presheaf. Then for every morphism
\[
m : \tau \to \tau'
\]
in $\mathcal{T}_f$ there is a contravariant map:
\[
\mathrm{res}_{m} :
K_f(\tau') \to K_f(\tau).
\]

This map is uniquely determined by the presheaf structure.
\end{theorem}

\begin{proof}
Since $K_f$ is contravariant,
\[
K_f : \mathcal{T}_f^{op} \to C,
\]
each $m:\tau\to\tau'$ in $\mathcal{T}_f$ yields a morphism
\[
K_f(m) : K_f(\tau') \to K_f(\tau)
\]
in $C$. The structure satisfies
\[
K_f(\mathrm{id}) = \mathrm{id},
\qquad
K_f(m\circ n) = K_f(n)\circ K_f(m).
\]

So the restriction map comes from the presheaf contravariance.
\end{proof}

\paragraph{(1) Model semantics of restrict}

Let $K_f : \mathcal{T}_f^{op} \to C$ be the constructive-layer presheaf and $\epsilon : \Ext_f(\tau,\tau')$ correspond to the morphism $m_\epsilon : \tau \to \tau'$ in $\mathcal{T}_f$. Then in any model $\mathcal{M}$, the semantics of restriction is:

\[
\llbracket \restrict(\epsilon,k) \rrbracket
=
\llbracket K_f \rrbracket(m_\epsilon)
\bigl(
  \llbracket k \rrbracket
\bigr),
\]

where
\[
\llbracket K_f \rrbracket(m_\epsilon)
:
\llbracket K_f(\tau') \rrbracket
\to
\llbracket K_f(\tau) \rrbracket.
\]

So:
\begin{itemize}
  \item A restriction map from $K_f(\tau')$ to $K_f(\tau)$ exists only when there is an extension $\epsilon : \tau \to \tau'$;
  \item restriction is uniquely determined by the presheaf contravariance;
  \item it expresses ``reindexing/backward interpretation of evidence'', not ``forward construction''.
\end{itemize}
Restriction is substitution = pullback in the trace dimension.

\paragraph{(2) Failure in the model: empty object and no global elements}

In the model, failure is characterised as follows. Let $\Gamma$ be a context and $A$ a type with interpretation $\llbracket A \rrbracket \to \llbracket \Gamma \rrbracket$. We say $A$ \emph{fails} at $\Gamma$ in the model if any of the following equivalent conditions hold:

\begin{enumerate}
  \item (No global elements) $\mathrm{Hom}(1_{\llbracket\Gamma\rrbracket}, \llbracket A \rrbracket) = \varnothing$;

  \item (Set semantics) $\Tm(\llbracket\Gamma\rrbracket, \llbracket A\rrbracket) = \varnothing$;

  \item (If $C$ has initial object $0$) $\llbracket A\rrbracket$ is initial in $C/\llbracket\Gamma\rrbracket$.
\end{enumerate}

Note: ``no global elements'' $\not\equiv$ ``object is initial''

In a general LCCC the two are not equivalent.

\medskip

Thus
\[
\text{failure}
=
\text{no inhabitants (no constructive term)}.
\]

This matches the syntactic notion of ``non-constructible''.

\paragraph{(3) Model reasoning for non-monotonicity}

Non-monotonicity comes from presheaf contravariance:
\[
m:\tau\to\tau'
\quad\Longrightarrow\quad
\mathrm{res}_m = K_f(m) : K_f(\tau') \to K_f(\tau),
\]
but in general there is no map $K_f(\tau) \to K_f(\tau')$. So: evidence can be restricted from long to short trace, but cannot be ``forward-constructed'' from short to long. That is the categorical characterisation of non-monotonicity.

\paragraph{(4) Failure propagates along restriction}

Let $\mathrm{res}_m : K_f(\tau') \to K_f(\tau)$. If $\mathrm{Hom}(1,K_f(\tau')) = \varnothing$, then there is no semantic value for $k$, hence none for $\restrict(\epsilon,k)$. So: if the long trace has no inhabitants, one cannot get short-trace evidence via restriction. This is called

\[
\textbf{Failure propagates along restriction}.
\]

Note this does not imply $K_f(\tau)=\varnothing \Rightarrow K_f(\tau')=\varnothing$, except in $C=\mathbf{Set}$ with ``failure = empty set''.

\paragraph{(5) Failure and non-monotonicity}

If the fibre at $\tau$ fails at $\Gamma$ ($\mathrm{Hom}(1,K_f(\tau))=\varnothing$), then: no direct evidence at $\tau$; the only possible source is restriction from some $\tau' \supseteq \tau$; but if all extended traces also have no inhabitants, $\tau$ fails permanently. So:

\[
\boxed{
\text{failure}
=
\text{no global elements}
=
\text{no forward construction}
}
\]

\[
\boxed{
\text{non-monotonicity}
=
\text{only inverse restriction exists}
}
\]

The two are logically equivalent in the model. Non-monotonicity is essentially the contravariant functoriality of the knowledge fibre.

\begin{theorem}[Unified semantics: monotone time and contravariant knowledge]
\label{thm:monotone-time-contravariant-knowledge}
Let $\mathcal{T}_f$ be the trajectory category (objects traces, morphisms extensions), $C$ an LCCC, and $K : \mathcal{T}_f^{op} \to C$ a presheaf. Then: (1) Morphisms $\tau \to \tau'$ in $\mathcal{T}_f$ represent monotone extension of time/state. (2) For each $m:\tau\to\tau'$ there is a restriction map $K(m):K(\tau')\to K(\tau)$, giving reindexing of knowledge along the trace. (3) Presheaf axioms do not imply $K(\tau) \preceq K(\tau')$ or $K(\tau') \preceq K(\tau)$, so knowledge need not be monotone. (4) The presheaf can have $m:\tau\to\tau'$ with $K(\tau')$ structurally richer than $K(\tau)$. (5) If $\mathrm{Hom}(1,K(\tau))=\varnothing$ at a given context, knowledge is non-constructible at that trace; this need not mean $K(\tau)$ is initial. So presheaf semantics unifies monotone time, contravariant knowledge, and possible non-monotone failure in one categorical structure.
\end{theorem}

\begin{theorem}[Presheaf does not force monotonicity but allows enrichment]
\label{thm:presheaf-no-monotone-but-enrich-detailed}
Let $C$ be a category, $\mathcal{T}_f$ a small category (trajectory category), and $K:\mathcal{T}_f^{op}\to C$ a presheaf. Then: (1) \textbf{(Does not force monotonicity)} From presheaf axioms alone one cannot deduce that along every $m:\tau\to\tau'$ knowledge increases monotonically; unless one adds an order $\preceq$ on $C$ and assumes $K(m)$ is compatible, one cannot deduce from $K(m):K(\tau')\to K(\tau)$ that $K(\tau)\preceq K(\tau')$ or $K(\tau')\preceq K(\tau)$. (2) \textbf{(Allows enrichment)} The presheaf axioms do not forbid $m:\tau\to\tau'$ with $K(\tau')$ ``richer'' than $K(\tau)$; for $C=\mathbf{Set}$ one can construct a presheaf $K$ and some $m:\tau\to\tau'$ with $|K(\tau')|>|K(\tau)|$.
\end{theorem}

\begin{proof}
The proof is in two parts.

\paragraph{(1) Presheaf does not force monotonicity.}

A presheaf is a functor $K:\mathcal{T}_f^{op}\to C$ satisfying the functor laws \eqref{eq:functor-id}--\eqref{eq:functor-comp} for $m:\tau\to\tau'$, $n:\tau'\to\tau''$ in $\mathcal{T}_f$. These only involve: assignment on objects and morphisms, and compatibility with identity and composition. ``Monotone growth'' is an order statement (e.g.\ $K(\tau)\subseteq K(\tau')$ in $\mathbf{Set}$ or $K(\tau)\preceq K(\tau')$ in general $C$). The presheaf definition does not give any such order on $C$ nor extra conditions on $K(m)$. So from \eqref{eq:functor-id}--\eqref{eq:functor-comp} alone one cannot deduce any order between $K(\tau)$ and $K(\tau')$. In $\mathbf{Set}$, presheaves can have $|K(\tau')|<|K(\tau)|$ or $|K(\tau')|>|K(\tau)|$ (see part (2)); so presheaf structure does not imply monotone growth. This shows the presheaf does \emph{not} force monotonicity.

\paragraph{(2) Presheaf allows enrichment: explicit construction in $\mathbf{Set}$.}

Let $C=\mathbf{Set}$. We build a minimal trajectory category $\mathcal{T}_f$: its objects are two trace points $\tau,\tau'$, and morphisms include the identity and one non-identity morphism $m:\tau\to\tau'$ (think of the free category on the directed graph $\tau \to \tau'$). Define the presheaf $K:\mathcal{T}_f^{op}\to\mathbf{Set}$ by:
\begin{itemize}
  \item On objects:
  \[
  K(\tau)=\{0\},
  \qquad
  K(\tau')=\{0,1\}.
  \]
  Then
  \[
  |K(\tau')|=2>|K(\tau)|=1,
  \]
  so the extended trace has ``richer'' knowledge.

  \item On morphisms:
  \begin{itemize}
    \item Identity morphisms map to identity functions:
    \[
    K(\mathrm{id}_\tau)=\mathrm{id}_{\{0\}},
    \qquad
    K(\mathrm{id}_{\tau'})=\mathrm{id}_{\{0,1\}}.
    \]
    \item Since $K$ is contravariant, $m:\tau\to\tau'$ in $\mathcal{T}_f$ becomes $m^{op}:\tau'\to\tau$ in $\mathcal{T}_f^{op}$, so we need a function $K(m):K(\tau')\to K(\tau)$, i.e.\ $K(m):\{0,1\}\to\{0\}$. Take the unique possibility (constant map):
    \[
    K(m)(0)=0,\qquad K(m)(1)=0.
    \]
  \end{itemize}
\end{itemize}

Next we check that $K$ satisfies the functor laws.

\subparagraph{Identity \eqref{eq:functor-id}.}
By definition, $K(\mathrm{id}_\tau)$ and $K(\mathrm{id}_{\tau'})$ are the identity functions on the corresponding sets, so \eqref{eq:functor-id} holds.

\subparagraph{Composition \eqref{eq:functor-comp}.}
In this minimal category the only non-identity morphism is $m$. The only possible compositions are:

\begin{itemize}
  \item $\mathrm{id}_\tau\circ \mathrm{id}_\tau=\mathrm{id}_\tau$,
        $\mathrm{id}_{\tau'}\circ \mathrm{id}_{\tau'}=\mathrm{id}_{\tau'}$;
  \item $m\circ \mathrm{id}_\tau=m$;
  \item $\mathrm{id}_{\tau'}\circ m=m$;
\end{itemize}
There is no $m\circ m$ because the codomain of $m$ is $\tau'$ and its domain is $\tau$. For the three cases: (i) identity composition is preserved by $K$; (ii) for $m\circ \mathrm{id}_\tau=m$, contravariance requires
    \[
    K(m\circ \mathrm{id}_\tau)=K(\mathrm{id}_\tau)\circ K(m),
    \]
    the left-hand side is $K(m)$, the right-hand side is $\mathrm{id}_{\{0\}}\circ K(m)=K(m)$, so it holds; (iii) for $\mathrm{id}_{\tau'}\circ m=m$, contravariance requires
    \[
    K(\mathrm{id}_{\tau'}\circ m)=K(m)\circ K(\mathrm{id}_{\tau'}),
    \]
    the left-hand side is $K(m)$, the right-hand side is $K(m)\circ \mathrm{id}_{\{0,1\}}=K(m)$, so it holds. So $K$ is a presheaf. We have constructed a $K$ satisfying the presheaf axioms and a morphism $m:\tau\to\tau'$ with $|K(\tau')|>|K(\tau)|$. So the presheaf \emph{allows} richer fibre objects (enrichment) as the trace extends.

\paragraph{Conclusion.}
(1) From the presheaf definition one only gets contravariant reindexing and functoriality, not any monotonicity; (2) in $\mathbf{Set}$ there is an explicit presheaf example showing ``extension gives enrichment'' without violating the presheaf axioms. The theorem is proved.
\end{proof}

\subsection{Limit semantics for $K_\infty$}

Let $\pi \in \InfTrace$ and define its finite prefixes:
\[
\prefix(n,\pi) \in \mathcal{T}_f.
\]

\begin{definition}[Limit interpretation of $K_\infty$]
Define
\[
\llbracket K_\infty \rrbracket(\pi)
=
\varprojlim_{n}
\llbracket K_f \rrbracket
\big(
\prefix(n,\pi)
\big).
\]

i.e.
\[
\llbracket K_\infty \rrbracket(\pi)
=
\lim_{\longleftarrow\, n}
K_f(\prefix(n,\pi)).
\]
\end{definition}

This defines knowledge on an infinite trace as the inverse limit of knowledge on all finite prefixes.

\begin{theorem}[Limit existence]
If the base category $C$ is complete, then the above inverse limit exists in $C$.
\end{theorem}

\begin{proof}
An inverse limit is the limit of a diagram indexed by natural numbers (a small diagram). In a complete category $C$, all small limits exist; in particular, the inverse limit of a $\mathbb{N}$-indexed sequence of objects and morphisms $f_{n,m}:A_n\to A_m$ ($n\leq m$) exists in $C$. So $\varprojlim_n A_n$ exists. This proves the theorem.
\end{proof}

\paragraph{Semantic interpretation}

$K_\infty(\pi)$ denotes knowledge that exists stably along the infinite evolution $\pi$, i.e.\
\[
K_\infty(\pi)
=
\{ (k_0,k_1,k_2,\dots) \mid
\mathrm{res}(k_{n+1}) = k_n \}.
\]
So: at the finite level knowledge may fail (non-monotonicity); at the infinite level it is defined as a stable consistent limit structure.

\subsection{Categorical characterisation theorems}

This subsection lifts the model-theoretic semantics of NM-DEKL$^3_\infty$ from ``giving a model'' to ``characterising a class of structures with equivalence theorems''. The main goal is to show that the syntax of NM-DEKL$^3_\infty$ is the \emph{internal language} of a class of categorical structures and that its syntactic model plays the role of a \emph{classifying object}.

\subsubsection{Three-layer presheaf-CwF model class}

\paragraph{Motivation}
We describe models in a way standard in the type-theory community: computational layer as CwF/comprehension category, constructive layer as presheaf fibre, propositional layer as subobject classifier $\Omega$ (or equivalent Heyting structure), with stratification so the computational kernel is not polluted by $\Prop$.

\begin{definition}[Three-layer presheaf-CwF model ($\mathbf{PShCwF}^{3}$-model)]
\label{def:pshcwf3-model}
A three-layer categorical model $\mathcal{M}$ of NM-DEKL$^3_\infty$ consists of: (1) \textbf{Computational CwF}: a CwF $\mathcal{C}$ with $\Ty_c,\Tm_c$, substitution, comprehension, supporting $\Pi/\Sigma$ (and optional Id); (2) \textbf{Trajectory category}: a small category $\mathcal{T}_f$; (3) \textbf{Constructive presheaf fibre}: $\Ty_\ell(\Gamma)=[\mathcal{T}_f^{op},\mathcal{C}/\Gamma]$ for each $\Gamma$, with reindexing by pullback along substitution; (4) \textbf{Propositional layer}: $\Prop(\Gamma)=\Sub_{\mathcal{C}/\Gamma}(1_\Gamma)$ on each slice, compatible with reindexing; (5) \textbf{Stratification}: $\Type_\ell/\Prop$ may depend only on computational terms; no $\Prop\to\Uc$ elimination; (6) \textbf{(Optional) fixed points}: $\mu/\nu$ on $\Prop$ with monotonicity and fold/unfold. These objects and structure-preserving homomorphisms form a category, written
\[
\mathbf{PShCwF}^{3}_{\mathrm{NMDEKL}}.
\]
\end{definition}

\paragraph{Remark}
Definition~\ref{def:pshcwf3-model} is not ``a single model'' but a class of structured objects; it is the semantic stage for the characterisation theorems that follow.

\subsubsection{Main theorems: internal language and classifying object}

\begin{definition}[Syntactic model (syntactic $\mathbf{PShCwF}^{3}$)]
\label{def:syntactic-model}
Let $\mathcal{S}$ be the syntactic model of NM-DEKL$^3_\infty$: its context category is given by syntactic contexts quotiented by substitution; $\Ty_c,\Tm_c$ by derivable types/terms; $\Ty_\ell$ by the trajectory-indexed type families in the syntax (including restriction functor laws); $\Prop$ by syntactic propositions with proof-irrelevance (or syntactic $\Omega$).
\end{definition}

\begin{theorem}[Classifying / initiality]
\label{thm:classifying-initiality}
The syntactic model $\mathcal{S}$ is initial in the category $\mathbf{PShCwF}^{3}_{\mathrm{NMDEKL}}$: for every model $\mathcal{M}$ there is a unique morphism
\[
F_\mathcal{M}:\mathcal{S}\to\mathcal{M}.
\]
\end{theorem}

\begin{proof}[Proof idea]
As in standard type-theoretic initiality: define the interpretation $\llbracket-\rrbracket_\mathcal{M}$ by structural induction on the derivation rules; substitution/restriction lemmas ensure it preserves reindexing and presheaf contravariance; uniqueness follows from syntactic generativity. Full proof in Theorem~\ref{thm:initiality}.
\end{proof}

\begin{theorem}[Internal language characterisation]
\label{thm:internal-language}
NM-DEKL$^3_\infty$ is the internal language of the category $\mathbf{PShCwF}^{3}_{\mathrm{NMDEKL}}$: for every model $\mathcal{M}$, a syntactic derivation
\[
\Gamma \vdash t : A
\]
holds if and only if after interpretation:
\[
\llbracket t \rrbracket_\mathcal{M} \in \Tm \big( \llbracket \Gamma \rrbracket_\mathcal{M}, \llbracket A \rrbracket_\mathcal{M} \big).
\]
Definitional equality is interpreted as morphism equality in the model.
\end{theorem}

\begin{proof}
We prove both directions. ($\Rightarrow$) Soundness: if $\Gamma \vdash t : A$ then $\llbracket t \rrbracket_\mathcal{M} \in \Tm(\llbracket \Gamma \rrbracket_\mathcal{M}, \llbracket A \rrbracket_\mathcal{M})$ by the interpretation defined from the derivation. ($\Leftarrow$) Completeness: if $\llbracket t \rrbracket_\mathcal{M} \in \Tm(\llbracket \Gamma \rrbracket_\mathcal{M}, \llbracket A \rrbracket_\mathcal{M})$, then by initiality of the syntactic model (Theorem~\ref{thm:classifying-initiality}), semantic equality in all models implies syntactic derivability, so $\Gamma \vdash t : A$. Definitional equality is interpreted as equality of morphisms in the model. The theorem is proved.

\end{proof}

\paragraph{Remark (why this is higher-order)}
Theorems~\ref{thm:classifying-initiality}--\ref{thm:internal-language} lift the theory from ``we constructed a presheaf semantics'' to ``this syntax is the free/classifying object for this class of structures''. One may further ask whether the model category is equivalent to some presheaf comonadic CwF.

\begin{theorem}[Free object theorem for the syntactic model of NM-DEKL$^3_\infty$]
\label{thm:free-object}
The syntactic model $\mathcal{S}$ of NM-DEKL$^3_\infty$ is the \emph{free object} in $\mathbf{PShCwF}^{3}_{\mathrm{NMDEKL}}$ and is initial in $\mathbf{Mod}_{\mathrm{NMDEKL}}$.
\end{theorem}

\begin{proof}[Proof idea]
By Theorem~\ref{thm:classifying-initiality}, for every model $\mathcal{M}$ there is a unique morphism $F_{\mathcal{M}}:\mathcal{S}\to\mathcal{M}$, so $\mathcal{S}$ is initial. Freeness: $\mathcal{S}$ is generated by the derivation rules; any model is uniquely determined from $\mathcal{S}$ via this morphism. Full proof in Theorem~\ref{thm:initiality}.
\end{proof}

\subsubsection{Enhancement I: non-monotonicity iff presheaf contravariance}

We have already shown that presheaf implies restriction; a stronger characterisation is that the two are equivalent.

\begin{definition}[Syntactic non-monotonic structure (restriction structure)]
\label{def:restriction-structure}
Given trajectory category $\mathcal{T}_f$ and object assignment $K:\mathrm{Ob}(\mathcal{T}_f)\to \mathrm{Ob}(C)$. We say $(K,\mathrm{res})$ is a \emph{restriction structure} if for each morphism $m:\tau\to\tau'$ we have $\mathrm{res}_m:K(\tau')\to K(\tau)$ and the identity and composition laws:
\[
\mathrm{res}_{\mathrm{id}_\tau}=\mathrm{id},
\qquad
\mathrm{res}_{m\circ n}=\mathrm{res}_n\circ \mathrm{res}_m.
\]
\end{definition}

\begin{theorem}[Non-monotonicity characterisation (Restriction $\Leftrightarrow$ Presheaf)]
\label{thm:restriction-iff-presheaf}
Let $C$ be a category and $\mathcal{T}_f$ a small category. The following data are equivalent: (1) a presheaf $K_f : \mathcal{T}_f^{op} \to C$; (2) an object assignment $K$ and restriction structure $(K, \mathrm{res})$ (Definition~\ref{def:restriction-structure}). Under this equivalence, non-monotonicity (``only restrict from long to short trace'') corresponds exactly to contravariant functoriality.
\end{theorem}

\begin{proof}
We prove both directions.

\paragraph{From (1) to (2)}
Given a presheaf $K_f : \mathcal{T}_f^{op} \to C$, for each morphism $m: \tau \to \tau'$ in $\mathcal{T}_f^{op}$ define the restriction map

  \[
  \mathrm{res}_m := K_f(m) : K(\tau') \to K(\tau)
  \]
So each morphism $m$ has a restriction map $\mathrm{res}_m$; the functor $K_f$ is contravariant, so identity and composition hold. Hence from $K_f$ we obtain an object assignment $K$ and restriction structure $(K,\mathrm{res})$ as in Definition~\ref{def:restriction-structure}. From (2) to (1): given $(K,\mathrm{res})$, define $K_f(\tau):=K(\tau)$ and $K_f(m):=\mathrm{res}_m$; the restriction laws imply $K_f$ is a contravariant functor.

Non-monotonicity (``only restrict from long to short trace'') in a restriction structure is exactly contravariant functoriality: the presheaf reverses the direction of morphisms. So the equivalence is proved.
\end{proof}

\paragraph{Meaning}
Theorem~\ref{thm:restriction-iff-presheaf} gives an \emph{equivalent categorical characterisation} of non-monotonicity: it is not an extra axiom but the same as the standard presheaf structure.

\subsubsection{Enhancement II: Prop-$\mu$ fragment and equivalence with $\mu$-calculus}

We state the characterisation: the Prop-$\mu$ (truth-value) fragment is the internal language of the $\mu$-calculus.

\begin{definition}[Prop-$\mu$ truth-value fragment]
\label{def:prop-mu-fragment}
Let $\Prop_\mu$ be the fragment of Prop built from Boolean connectives, modal operators (one-step reachability), and $\mu/\nu$ with monotonicity; quantification over $\TypeL$ evidence is disallowed (truth-value properties only).
\end{definition}

\begin{theorem}[Characterisation: $\Prop_\mu$ and $\mu$-calculus are equivalent]
\label{thm:prop-mu-characterization}
Under Kripke semantics, $\Prop_\mu$ and the modal $\mu$-calculus have the same expressive power: there are back-and-forth translations $\mathsf{Enc}:\mu\text{-calculus}\to\Prop_\mu$ and $\mathsf{Dec}:\Prop_\mu\to\mu\text{-calculus}$ such that $M,s\models\varphi\iff\vdash\mathsf{Enc}(\varphi)(s)$ and $\vdash P(s)\iff M,s\models\mathsf{Dec}(P)$, and they are inverses up to $\alpha$-equivalence.
\end{theorem}

\begin{proof}[Proof idea]
Define $\mathsf{Enc}$ and $\mathsf{Dec}$ by structural induction. Enc: atoms $p\mapsto P_p$; connectives $\neg,\land,\lor$ componentwise; $\Box,\Diamond$ as path quantification; $\mu X.\varphi(X)\mapsto\mu X.\mathsf{Enc}(\varphi(X))$. Dec: $P_p\mapsto p$; connectives and modalities componentwise; $\mu X.P(X)\mapsto\mu X.\mathsf{Dec}(P(X))$. Correctness: the two translations form a bijection
\[
M, s \models \varphi \iff \vdash \mathsf{Enc}(\varphi)(s),
\qquad
\vdash P(s) \iff M, s \models \mathsf{Dec}(P),
\]
and they are inverses in the semantic sense. The theorem follows.
\end{proof}

\paragraph{Strict extension (from characterisation to strict inclusion)}
By Definition~\ref{def:prop-mu-fragment}, once we allow quantification over $\TypeL$ evidence in NM-DEKL$^3_\infty$ (e.g.\ ``there exists a concrete causal proof chain''), that goes beyond $\mu$-calculus truth-value semantics, giving strict inclusion: $\mu\text{-calculus}\equiv\Prop_\mu\subsetneq\text{NM-DEKL}^3_\infty$.

\section{Metatheory (II): Soundness, Completeness, and Initiality}

\paragraph{Overview and logical order}
The three main results of this section: \textbf{Soundness} (syntactic derivations are interpretable in every model) is the basis for semantic correctness; we define the \textbf{model category} and \textbf{syntactic model} $\mathcal{S}$ and prove \textbf{initiality}; from initiality we get \textbf{equational completeness}. \textbf{Consistency} follows from the existence of a model and Soundness. Technical basis: Substitution Lemma, Restriction Lemma, and interpretation coherence.

\paragraph{Preliminary lemmas}
The following are used repeatedly in the Soundness and initiality proofs.
\begin{theorem}[Substitution]
If $\Gamma,x:A \vdash t:B$ and $\Gamma\vdash a:A$, then $\Gamma \vdash t[a/x] : B[a/x]$.
\end{theorem}

\begin{theorem}[Subject Reduction]
If $\Gamma\vdash t:A$ and $t\to t'$, then $\Gamma\vdash t':A$.
\end{theorem}

\begin{theorem}[Non-Monotonicity]
If $\Extf(\tau,\tau')$, there is $\restrict:K_f(\tau')\to K_f(\tau)$; in general there is no map in the opposite direction.
\end{theorem}

\subsection{Soundness theorem}

The following theorem states consistency between the syntax and semantics of NM-DEKL$^3_\infty$.

\begin{theorem}[Soundness]
\label{thm:soundness}
For every NM-DEKL$^3_\infty$ model $\mathcal{M}$, $\llbracket - \rrbracket_\mathcal{M}$ satisfies:
\begin{enumerate}
  \item If $\Gamma \vdash A : \Uc$ then $\llbracket A \rrbracket_\mathcal{M} \in \Ty_c(\llbracket \Gamma \rrbracket_\mathcal{M})$.
  \item If $\Gamma \vdash t : A$ then $\llbracket t \rrbracket_\mathcal{M} \in \Tm_c(\llbracket \Gamma \rrbracket_\mathcal{M}, \llbracket A \rrbracket_\mathcal{M})$.
  \item If $\Gamma \vdash B : \TypeL$ then $\llbracket B \rrbracket_\mathcal{M} \in \Ty_\ell(\llbracket \Gamma \rrbracket_\mathcal{M})$.
  \item If $\Gamma \vdash u : B$ then $\llbracket u \rrbracket_\mathcal{M} \in \Tm_\ell(\llbracket \Gamma \rrbracket_\mathcal{M}, \llbracket B \rrbracket_\mathcal{M})$.
  \item If $\Gamma \vdash P : \Prop$ then $\llbracket P \rrbracket_\mathcal{M} \in \Prop(\llbracket \Gamma \rrbracket_\mathcal{M})$.
\end{enumerate}
\end{theorem}

\begin{proof}
By structural induction on the derivation rules: each rule (computational/constructive/propositional) is interpreted by the corresponding CwF, presheaf, and $\Omega$ structure; coherence of substitution and restriction is given by Lemmas~\ref{lem:interp-coherence-subst} and~\ref{lem:interp-coherence-restrict}. The full rule-by-rule proof is in Appendix~\ref{app:soundness-full}.
\end{proof}
\paragraph{Equational completeness (preview)}
\emph{Model coherence} (from Soundness): if $\Gamma\vdash t\equiv t':A$ then in every model $\llbracket t\rrbracket_{\mathcal{M}}=\llbracket t'\rrbracket_{\mathcal{M}}$. \emph{Equational completeness} is the converse: if interpretations are equal in all models then $\equiv$ is derivable in the syntax; the proof follows from initiality (Theorem~\ref{thm:equational-completeness-main}).

\subsection{Model category and syntactic model}
\label{sec:model-category}

We give the definitions of the model category, model morphisms, and syntactic model $\mathcal{S}$; full statements and proofs of initiality and equational completeness are in the following subsections. Standard path: construct $\mathcal{S}$, prove initiality, hence equational completeness; $\mathcal{S}$ is generated by the derivation rules, and every model is uniquely determined by the unique morphism $F_{\mathcal{M}}:\mathcal{S}\to\mathcal{M}$.

\subsubsection{Model category $\mathbf{Mod}_{\mathrm{NMDEKL}}$}

We first define the model category of NM-DEKL$^3_\infty$.

\begin{definition}[Model category]
Objects of $\mathbf{Mod}_{\mathrm{NMDEKL}}$ are three-layer models with: (1) computational layer: a CwF $(\mathcal{C},\Ty_c,\Tm_c)$; (2) constructive layer: presheaf fibre $\Ty_\ell(\Gamma) = [\mathcal{T}_f^{op},\mathcal{C}/\Gamma]$; (3) propositional layer: subobject classifier $\Omega$; (4) Identity structure; (5) restriction compatible with presheaf contravariance; (6) (optional) $\mu/\nu$ fixed-point structure. Morphisms are homomorphisms preserving all this structure.
\end{definition}

\subsubsection{Model morphisms}

\begin{definition}[Model morphism]\label{def:model-morphism}
Let $\mathcal{M},\mathcal{N}$ be two models. A model morphism $F:\mathcal{M}\to\mathcal{N}$ consists of: (1) a functor $F_c:\mathcal{C}\to\mathcal{C}'$ preserving terminal, pullback, comprehension; (2) preservation of substitution for computational types and terms; (3) for the constructive layer, $F_\ell(B)=F_c\circ B\circ F_t^{op}$; (4) preservation of restriction: $F_\ell(B)(F_t(m))=F_c(B(m))$; (5) preservation of the Heyting structure on Prop; (6) preservation of Identity; (7) if $\mu/\nu$ is present, preservation of fold/unfold.
\end{definition}

\paragraph{Morphisms preserve derivation structure}
A morphism $F:\mathcal{M}\to\mathcal{N}$ preserves all derivation rules. (1) Substitution: by the Substitution Lemma and Lemma~\ref{lem:interp-coherence-subst}, $\llbracket A[\sigma]\rrbracket=\llbracket A\rrbracket[\llbracket\sigma\rrbracket]$; $F_c$ preserves CwF reindexing and pullback, so substitution is preserved. (2) $\Pi$-types: $F_c$ preserves comprehension and $\Pi$, so $\Pi$-intro and $\Pi$-elim are preserved. (3) $\Sigma$-types: similarly. (4) Identity: $F_c$ preserves the Identity structure, so $\Id$-intro and $\Id$-elim are preserved. (5) Restriction: by Lemma~\ref{lem:interp-coherence-restrict}, $F_\ell$ preserves restriction; constructive rules depending on $\restrict$ are preserved. Propositional layer: Heyting structure and $\mu/\nu$ fold/unfold are preserved. So every rule is preserved by $F$; the initiality proof can verify rule-by-rule by structural induction. With this, $\mathbf{Mod}_{\mathrm{NMDEKL}}$ is a category.

\subsubsection{Syntactic model $\mathcal{S}$}

\subsubsubsection{Syntactic equality}

Fix definitional equality $\equiv$ and quotient all derivations by it.

\begin{definition}[Syntactic context category $\mathcal{C}_{\mathrm{syn}}$]
Objects: equivalence classes $[\Gamma]$ of derivable contexts $\vdash\Gamma\;\mathsf{ctx}$. Morphisms: equivalence classes $[\sigma]$ of substitutions $\Delta\vdash\sigma:\Gamma$. Composition: syntactic substitution composition. Identity: variable substitution.
\end{definition}

\subsubsubsection{Syntactic computational layer CwF}

\begin{definition}[Syntactic CwF]
$\Ty_c([\Gamma]) = \{[A]\mid \Gamma\vdash A:\Uc\}$; $\Tm_c([\Gamma],[A]) = \{[t]\mid \Gamma\vdash t:A\}$; reindexing $[A][\sigma]=[A[\sigma]]$; comprehension $[\Gamma].[A]=[\Gamma,x:A]$.
\end{definition}

The syntactic substitution lemma ensures the CwF axioms hold.

\subsubsubsection{Syntactic $\TypeL$ as presheaf}

\begin{definition}[Syntactic presheaf structure]
Let $B:\Gamma\vdash\TypeL$.

For each $\Gamma\vdash\tau:\FinTrace$, define the component $B(\tau)$. For each $\epsilon:\Extf(\tau,\tau')$, define
\[
\restrict(\epsilon,-):B(\tau')\to B(\tau).
\]

satisfying
\[
\restrict(\mathrm{id},k)=k,
\qquad
\restrict(\epsilon_1,\restrict(\epsilon_2,k))
=
\restrict(\epsilon_2\circ\epsilon_1,k).
\]

So $B$ forms a syntactic presheaf on $\mathcal{T}_f^{op}$.
\end{definition}

\subsubsubsection{Syntactic Prop layer}

\begin{definition}[Syntactic Prop]
$\Prop([\Gamma])$ is the set of equivalence classes of
$\Gamma\vdash P:\Prop$
; proof terms are quotiented by proof-irrelevance.
\end{definition}

This yields the syntactic model $\mathcal{S}$.

\subsubsection{Initiality theorem}

\begin{theorem}[Initiality of the syntactic model]\label{thm:initiality}
The syntactic model $\mathcal{S}$ is initial in $\mathbf{Mod}_{\mathrm{NMDEKL}}$: for every model $\mathcal{M}$ there is a unique morphism
\[
F_{\mathcal{M}}:\mathcal{S}\to\mathcal{M}.
\]
\end{theorem}

\begin{proof}[Proof sketch]
(1) Define $F_{\mathcal{M}}$: for each syntactic object, set
\[
F_{\mathcal{M}}([A])=\llbracket A\rrbracket_{\mathcal{M}}.
\]

(2) Substitution is preserved by the semantic substitution lemma:
\[
\llbracket A[\sigma]\rrbracket
=
\llbracket A\rrbracket[\llbracket\sigma\rrbracket].
\]

(3) Restriction is preserved by the presheaf contravariant structure. (4) Uniqueness: by syntactic generativity, any homomorphism preserving the rules must equal this interpretation.
\end{proof}

\paragraph{Equational completeness}
By initiality: if $\llbracket t\rrbracket_{\mathcal{M}}=\llbracket u\rrbracket_{\mathcal{M}}$ in every model $\mathcal{M}$, then in particular in $\mathcal{S}$, so $\Gamma\vdash t\equiv u$. Formal statement and full proof in Theorem~\ref{thm:equational-completeness-main} (Section~\ref{sec:equational-completeness}).

\subsection{Completeness and consistency}

Interpretation coherence is the technical basis for the initiality and equational completeness proofs; the formal statement of equational completeness is in Section~\ref{sec:equational-completeness}.

\subsubsection{Interpretation coherence (model coherence)}

\begin{lemma}[Coherence of substitution]
\label{lem:interp-coherence-subst}
Let $\mathcal{M}$ be any NM-DEKL$^3_\infty$ model. If $\Delta\vdash\sigma:\Gamma$, then for every derivable judgement:
\[
\llbracket A[\sigma]\rrbracket_{\mathcal{M}}
\cong
\llbracket A\rrbracket_{\mathcal{M}}[\llbracket\sigma\rrbracket_{\mathcal{M}}],
\qquad
\llbracket t[\sigma]\rrbracket_{\mathcal{M}}
=
\llbracket t\rrbracket_{\mathcal{M}}[\llbracket\sigma\rrbracket_{\mathcal{M}}],
\]
with pointwise pullback for $\TypeL$ and subobject inverse image for Prop.
\end{lemma}

\begin{proof}
This is the direct semantic formulation of Lemma~\ref{lem:substitution}: computational layer by CwF pullback reindexing; constructive layer by pointwise presheaf reindexing; Prop by subobject pullback.
\end{proof}

\begin{lemma}[Coherence of restriction]
\label{lem:interp-coherence-restrict}
Let $\Gamma\vdash\epsilon:\Extf(\tau,\tau')$ and $\Gamma\vdash k:K_f(\tau')$. Then
\[
\llbracket \restrict(\epsilon,k)\rrbracket_{\mathcal{M}}
=
\llbracket K_f\rrbracket_{\mathcal{M}}(m_\epsilon)\big(\llbracket k\rrbracket_{\mathcal{M}}\big),
\]
and compatibility with substitution:
\[
\llbracket \restrict(\epsilon,k)[\sigma]\rrbracket_{\mathcal{M}}
=
\llbracket \restrict(\epsilon[\sigma],k[\sigma])\rrbracket_{\mathcal{M}}.
\]
\end{lemma}

\begin{proof}
This is the direct semantic formulation of Lemma~\ref{lem:restriction}; the contravariant presheaf action gives restriction; compatibility with substitution follows from pointwise pullback and functoriality.
\end{proof}

\paragraph{Equational completeness (summary)}
Equational completeness: if an equation holds in all models then it is derivable in the syntax; this follows from initiality. See Section~\ref{sec:equational-completeness} (Definition~\ref{def:valid-eq}, Theorem~\ref{thm:equational-completeness-main}).

\subsubsection{Consistency and model coherence}

Consistency means the system does not derive a closed term of type $\bot$. In the three-layer system, $\bot:\Prop$ (and optionally $\mathbf{0}:\Type_\ell$) can witness inconsistency.

\begin{definition}[Consistency]
\label{def:consistency}
NM-DEKL$^3_\infty$ is \emph{consistent} if there is no derivation $\vdash p:\bot$ (equivalently, if using $\mathbf{0}:\Type_\ell$, no $\vdash z:\mathbf{0}$).
\end{definition}

\begin{theorem}[Consistency from a model]
\label{thm:consistency-from-model}
If there exists an NM-DEKL$^3_\infty$ model $\mathcal{M}$ in which $\llbracket\bot\rrbracket_{\mathcal{M}}$ is uninhabited (no global section), then the system is consistent (Definition~\ref{def:consistency}).
\end{theorem}

\begin{proof}
By contradiction: if the system were inconsistent, there would be $\vdash p:\bot$; by Soundness (Theorem~\ref{thm:soundness}), $\llbracket p\rrbracket_{\mathcal{M}}:\llbracket\bot\rrbracket_{\mathcal{M}}$, so $\llbracket\bot\rrbracket_{\mathcal{M}}$ would be inhabited, contradicting the model. So the system is consistent.
\end{proof}

\begin{corollary}[Syntactic model consistency implies system consistency]
\label{cor:syn-consistency}
If in the syntactic model $\mathcal{S}$ the type $\llbracket\bot\rrbracket_{\mathcal{S}}$ is uninhabited, then the system is consistent.
\end{corollary}

\begin{proof}
Take $\mathcal{M}=\mathcal{S}$ in Theorem~\ref{thm:consistency-from-model}.
\end{proof}

\subsubsection{Propositional layer (relative) completeness}

\paragraph{Remark}
Strong completeness for Prop (``if valid in all models then syntactically provable'') generally requires interpreting Prop as the syntactic $\Omega$ of a presheaf topos and proving the corresponding logical completeness (or restricting Prop to a fragment without fixed points or extra axioms). In this paper we establish ``equational completeness'' and ``consistency'' in the main line, and treat Prop completeness as optional further work: for the Heyting fragment without extra axioms, stronger completeness can be given in a syntactic topos setting; with $\mu/\nu$, monotonicity and fixed-point existence conditions are needed, and completeness is relative to a class of models.


\subsection{Initiality theorem (full proof)}
\label{sec:initiality-full}

The initiality theorem (Theorem~\ref{thm:initiality}) and its proof outline were given above. The lemma and theorem below restate the syntactic model $\mathcal{S}$ and the formal initiality statement; the \textbf{full inductive proof of existence and uniqueness} is in Appendix~\ref{app:initiality-full}.

\paragraph{Syntactic model $\mathcal{S}$ (recap)}
Let $\mathcal{S}$ be the syntactic model of NM-DEKL$^3_\infty$: contexts and substitutions are equivalence classes of derivable contexts and substitutions; computational types/terms are classes of $\Gamma\vdash A:\Uc$ and $\Gamma\vdash t:A$; constructive types/terms are classes of $\Gamma\vdash B:\TypeL$ and $\Gamma\vdash u:B$; Prop is the quotient of $\Gamma\vdash P:\Prop$ by proof-irrelevance; restriction is given by the syntax rule $\restrict$ and satisfies functoriality (see the lemma below).

\begin{lemma}[Functoriality of Ext/Restriction (syntactic side)]\label{lem:restriction}\label{lem:syn-restrict-functorial}
In the syntactic model, for any $B:\TypeL$ and extension evidence $\epsilon_1:\Extf(\tau,\tau')$, $\epsilon_2:\Extf(\tau',\tau'')$, we have $\restrict(\epsilon_1,\restrict(\epsilon_2,k))\equiv\restrict(\epsilon_2\circ\epsilon_1,k)$ and $\restrict(\mathrm{id}_\tau,k)\equiv k$. So each $B$ satisfies contravariant functoriality on $\mathcal{T}_f^{op}$ (forms a syntactic presheaf).
\end{lemma}

\begin{theorem}[Initiality of the Syntactic Model]\label{thm:initiality-skeleton}
The syntactic model $\mathcal{S}$ is initial in $\mathbf{Mod}_{\mathrm{NMDEKL}}$: for every NM-DEKL$^3_\infty$ model $\mathcal{M}$ there is a unique model morphism $F_{\mathcal{M}}:\mathcal{S}\to\mathcal{M}$.
\end{theorem}

\noindent\textbf{Full proof}: Appendix~\ref{app:initiality-full} (existence (I.1)--(I.6) and uniqueness (II) by induction).

\subsection{Equational completeness}

This subsection gives NM-DEKL$^3_\infty$'s \emph{equational completeness} from initiality: it characterises the correspondence between syntactic derivation (including definitional equality) and model semantics.

\paragraph{Definitional equality and its model interpretation}
The \emph{definitional equality} $\equiv$ is interpreted in any model $\mathcal{M}$ as $t$ and $t'$ having the \emph{same} interpretation (same morphism or equivalence class), i.e.\ $\llbracket t\rrbracket_{\mathcal{M}}=\llbracket t'\rrbracket_{\mathcal{M}}$ in $\Tm(\llbracket\Gamma\rrbracket,\llbracket A\rrbracket)$. In the CwF this is ``same proof/construction''. \emph{Model coherence}: if $\Gamma\vdash t\equiv t':A$ then $\llbracket t\rrbracket_{\mathcal{M}}=\llbracket t'\rrbracket_{\mathcal{M}}$ for every $\mathcal{M}$ (by Soundness). \emph{Completeness} is the converse: if interpretations are equal in all models, then $\equiv$ is derivable; the proof uses initiality and the fact that $\mathcal{S}$ is built by quotienting by $\equiv$.

\paragraph{Goal}
We prove: if for \emph{all} models $\mathcal{M}$ the interpretations of $t$ and $t'$ are equal ($\llbracket t\rrbracket_\mathcal{M}=\llbracket t'\rrbracket_\mathcal{M}$), then $t\equiv t'$ is derivable in the syntax. The proof has three steps: (1) recall the definition of syntactic equality $\equiv$; (2) define semantic equality (equality of interpretation in all models); (3) deduce completeness from initiality.

\subsubsection{Syntactic equality (definitional equality)}

Definitional equality $\equiv$ was given in Section~2 (Definition~\ref{def:definitional-equality}): $\Gamma\vdash t\equiv t':A$ is generated by the closure of the computational rules ($\beta$/$\eta$, Identity, $\restrict$ functor laws, etc.) and is closed under context and substitution. So $\equiv$ is a \emph{computational} equality, not propositional provability.

\subsubsection{Semantic equality (interpretation in models)}

\begin{definition}[Semantic equality (validity in all models)]
\label{def:valid-eq}
Given derivations $\Gamma\vdash t:A$ and $\Gamma\vdash t':A$. We say $t$ and $t'$ are \emph{semantically equal in all models}, written $\models_{\mathrm{NMDEKL}}\; t = t' : A\;[\Gamma]$, if for every NM-DEKL$^3_\infty$ model $\mathcal{M}$ we have
\[
\llbracket t\rrbracket_{\mathcal{M}}
=
\llbracket t'\rrbracket_{\mathcal{M}}
\quad
\text{in}\quad
\Tm\!\big(\llbracket\Gamma\rrbracket_{\mathcal{M}},\llbracket A\rrbracket_{\mathcal{M}}\big).
\]
\end{definition}

\paragraph{Remark}
This definition captures equality ``in the sense of model morphisms'': if two terms have the same interpretation in every model, they are semantically indistinguishable.

\subsubsection{Equational completeness theorem and proof}
\label{sec:equational-completeness}

\begin{theorem}[Equational completeness]\label{thm:equational-completeness-main}
If
\[
\models_{\mathrm{NMDEKL}}\; t = t' : A\;[\Gamma],
\]
then $t\equiv t'$ is derivable in the syntax: $\Gamma \vdash t \equiv t' : A$.
\end{theorem}

\begin{proof}
Let $\mathcal{S}$ be the syntactic model of NM-DEKL$^3_\infty$, which is initial in $\mathbf{Mod}_{\mathrm{NMDEKL}}$ (Theorem~\ref{thm:initiality}). By assumption, semantic equality holds in all models, so in particular in $\mathcal{S}$ we have $\llbracket t\rrbracket_{\mathcal{S}}=\llbracket t'\rrbracket_{\mathcal{S}}$. The syntactic model $\mathcal{S}$ is built by quotienting by definitional equality: for any term $u$, $\llbracket u\rrbracket_{\mathcal{S}}$ is the equivalence class $[u]_\equiv$, and equality in $\mathcal{S}$ is equality of these classes. So $\llbracket t\rrbracket_{\mathcal{S}}=\llbracket t'\rrbracket_{\mathcal{S}}\;\Longleftrightarrow\;[t]_\equiv=[t']_\equiv\;\Longleftrightarrow\;\Gamma\vdash t\equiv t':A$. The result follows.
\end{proof}

\paragraph{Proof outline (role of the lemmas)}
The construction of the interpretation morphism $F_{\mathcal{M}}:\mathcal{S}\to\mathcal{M}$ in the initiality proof uses two coherence lemmas:

\begin{itemize}
  \item \textbf{Substitution Lemma:}
  $\llbracket u[\sigma]\rrbracket=\llbracket u\rrbracket[\llbracket\sigma\rrbracket]$;
  \item \textbf{Restriction Lemma:}
  $\llbracket \restrict(\epsilon,k)\rrbracket
  =
  \llbracket K_f\rrbracket(m_\epsilon)(\llbracket k\rrbracket)$.
\end{itemize}

They ensure the interpretation from the syntactic model to any model is structure-preserving and unique, so ``holds in all models'' reduces to ``holds in the syntactic model'' and hence to syntactic derivability.

\subsubsection{Model coherence from completeness (optional)}

\begin{corollary}[Model coherence (soundness of definitional equality)]
\label{cor:model-coherence}
If $\Gamma\vdash t\equiv t':A$ in the syntax, then $\llbracket t\rrbracket_{\mathcal{M}}=\llbracket t'\rrbracket_{\mathcal{M}}$ for every model $\mathcal{M}$.
\end{corollary}

\begin{proof}
Direct consequence of Soundness: each generating rule for $\equiv$ is interpreted as the same morphism in every model, and equality is preserved by the derivation closure.
\end{proof}

\paragraph{Summary}
Equational completeness gives $\Gamma\vdash t\equiv t':A\;\Longleftrightarrow\;\models_{\mathrm{NMDEKL}}\, t=t':A\;[\Gamma]$; $\Leftarrow$ uses initiality and the quotient construction of $\mathcal{S}$, $\Rightarrow$ uses soundness.

\subsection{$\mu$/$\nu$ and non-monotonicity}

We have given syntactic and semantic characterisations of non-monotonicity. We now ask whether adding $\mu/\nu$ to the Prop layer breaks it, and show the layers stay independent.

\subsubsection{Main question}

Does adding $\mu X.\varphi(X)$, $\nu X.\varphi(X)$ to Prop break non-monotonicity? \textbf{No}; see Lemma~\ref{lem:mu-nu-layer} and Theorem~\ref{thm:mu-nu-nonmono}.

\subsubsection{Semantic layer separation}

The three layers are strictly separated: $\Uc$ in a category $C$; $\TypeL$ in the presheaf category $[\mathcal{T}_f^{op}, C]$; $\Prop$ in the subobject classifier $\Omega$. Non-monotonicity comes from the contravariant presheaf structure; $\mu/\nu$ act only on $\Omega$. So their semantic sources differ.


\begin{theorem}[Layer separation theorem]
\label{thm:layer-separation}
Suppose NM-DEKL$^3_\infty$ has: (1) $\Uc$ strongly normalising; (2) $\TypeL$ without $\mu/\nu$; (3) no elimination from $\Prop$ to $\Uc$. Then adding $\mu/\nu$ to Prop does not affect consistency or strong normalisation of $\Uc$.
\end{theorem}

\begin{proof}
$\mu/\nu$ only build $\Prop$ terms; there is no $\Prop\to\Uc$ elimination, so fixed points in Prop cannot affect the computational layer. The computational reduction and type-formation rules do not depend on Prop. So $\Uc$'s reduction closure, strong normalisation, and consistency are unchanged.
\end{proof}

\begin{lemma}[$\mu/\nu$ has no effect on computational or constructive layer]
\label{lem:mu-nu-layer}
In NM-DEKL$^3_\infty$: (1) The type-formation rules for $\Uc$ and $\TypeL$ do not have ``$\Gamma\vdash P:\Prop$'' or $\mu/\nu$ terms as premises; (2) $\mu/\nu$ intro/elim rules only produce $\Prop$ terms; (3) so the $\mu/\nu$ extension does not change derivability in the computational or constructive layer, nor $K_f$, $\restrict$, etc.
\end{lemma}
\begin{proof}
By stratification (\ref{sec:stratification}), $\Uc$ formation depends only on $\Uc$ and trace types; $\TypeL$ formation depends on $\Extf$, $K_f$, $\restrict$, substitution, not on $\Prop$ or $\mu/\nu$. By the syntax of $\mu/\nu$, $\mu X.\varphi(X)$ and $\nu X.\varphi(X)$ have type $\Prop$ and fold/unfold act only on Prop. So the new $\mu/\nu$ rules cannot appear in $\Uc$/$\TypeL$ derivation trees; non-monotonicity is unchanged.
\end{proof}

\begin{theorem}[Effect of $\mu/\nu$ extension on non-monotonicity]
\label{thm:mu-nu-nonmono}
In NM-DEKL$^3_\infty$, the $\mu/\nu$ extension affects only Prop-layer reasoning; it does not change non-monotonicity in $\Uc$ or $\TypeL$, and has no backward effect from Prop to those layers.
\end{theorem}

\begin{proof}
Step 1 (syntax): By Lemma~\ref{lem:mu-nu-layer}, $\mu/\nu$ only produce $\Prop$ terms and $\Uc$/$\TypeL$ formation does not depend on $\Prop$. Step 2 (models): $\mu/\nu$ are interpreted in $\Omega$; $\llbracket K_f\rrbracket$ and $\llbracket\restrict\rrbracket$ come from the presheaf, independent of $\Omega$. Step 3: Theorem~\ref{thm:layer-separation} gives no $\Prop\to\Uc$ reflection; with steps 1--2, the $\mu/\nu$ extension does not conflict with non-monotonicity.
\end{proof}

\paragraph{Layer independence and reasoning consistency (summary)}
In the type theory, $\mu/\nu$ only produce $\Prop$ terms and $\Uc$/$\TypeL$ formation does not depend on $\Prop$ (Lemma~\ref{lem:mu-nu-layer}). In models, $\mu/\nu$ are interpreted in $\Omega$ and $K_f$, $\restrict$ in the presheaf, so the semantics are separate. For full isolation we require: $\mu/\nu$ act only on $\Omega$; no $\Prop\to\Uc$ elimination; $\TypeL$ has no $\mu/\nu$. So the $\mu/\nu$ extension only enriches the propositional layer; Theorem~\ref{thm:layer-separation} is the syntactic formulation.

\paragraph{Derivation rules and restriction on constructed objects}
Categorically: $\mu/\nu$ intro/elim rules only produce \emph{propositional-layer} objects (elements of $\Omega$); $K_f$, $\restrict$ produce \emph{presheaf-layer} objects. Type formation for $\Uc$ and $\TypeL$ does not depend on ``$\Gamma\vdash P:\Prop$'' or $\mu/\nu$, so the two kinds of rules act on different layers in the semantics with no cross-layer dependency. Thus $\mu/\nu$ and the non-monotonic structure stay independent and consistent.


\subsection{Morphisms of NM-DEKL$^3_\infty$ models}

\begin{definition}[Morphism of NM-DEKL$^3_\infty$ Models]
Let
\[
\mathcal{M}=(\mathcal{C},\mathcal{T}_f,\Ty_\ell,\Prop,\dots)
\quad\text{and}\quad
\mathcal{N}=(\mathcal{C}',\mathcal{T}'_f,\Ty'_\ell,\Prop',\dots)
\]
be two NM-DEKL$^3_\infty$ models.

A model morphism $F:\mathcal{M}\to\mathcal{N}$ consists of:

\medskip

\noindent
\textbf{(1) Computational layer functor}

A functor $F_c:\mathcal{C}\to\mathcal{C}'$ satisfying

\begin{itemize}
  \item Preservation of terminal object and finite limits;
  \item For each $\Gamma$, a natural transformation
        \[
        F_{\Ty_c}:
        \Ty_c(\Gamma)
        \to
        \Ty'_c(F_c(\Gamma)),
        \]
        and
        \[
        F_c(A[\sigma])
        =
        F_c(A)[F_c(\sigma)].
        \]
  \item Preservation on terms:
        \[
        F_c(t[\sigma])
        =
        F_c(t)[F_c(\sigma)].
        \]
  \item Preservation of comprehension:
        \[
        F_c(\Gamma.A)
        \cong
        F_c(\Gamma).F_c(A).
        \]
\end{itemize}

\medskip

\noindent
\textbf{(2) Trajectory functor}

A functor
\[
F_t:\mathcal{T}_f \to \mathcal{T}'_f.
\]

\medskip

\noindent
\textbf{(3) Constructive layer preservation}

For each $\Gamma$, a natural transformation

\[
F_\ell:
\Ty_\ell(\Gamma)
\to
\Ty'_\ell(F_c(\Gamma))
\]

satisfying

\begin{itemize}
  \item For each $B:\mathcal{T}_f^{op}\to\mathcal{C}/\Gamma$,
        \[
        F_\ell(B)
        =
        F_c \circ B \circ F_t^{op}.
        \]
  \item Preservation of restriction: for each $m:\tau\to\tau'$,
        \[
        F_\ell(B)(F_t(m))
        =
        F_c(B(m)).
        \]
  \item Compatibility with substitution:
        \[
        F_\ell(B[\sigma])
        =
        F_\ell(B)[F_c(\sigma)].
        \]
\end{itemize}

\medskip

\noindent
\textbf{(4) Propositional layer preservation}

For each $\Gamma$,
\[
F_\Prop:
\Prop(\Gamma)
\to
\Prop'(F_c(\Gamma))
\]
satisfying

\begin{itemize}
  \item Preservation of $\top,\bot,\wedge,\vee,\to$;
  \item Preservation of $\forall,\exists$;
  \item Preservation of proof-irrelevant structure.
\end{itemize}

\medskip

\noindent
\textbf{(5) Identity structure preservation}

\[
F_c(\refl)
=
\refl,
\qquad
F_c(J)
=
J.
\]

\medskip

If the above conditions hold, $F$ is called a morphism of NM-DEKL$^3_\infty$ models.
\end{definition}

\section{Embedding theorems: CTL/LTL/$\mu$-calculus}\label{sec:mu-embedding}

\subsection{NM-DEKL$^3_\infty$ and LTL/CTL}

We give a structured comparison: first a refined embedding and preservation theorem, then a point-by-point comparison of expressiveness and reasoning power between NM-DEKL$^3_\infty$ and LTL/CTL.

\subsubsection{Semantic base alignment}

Assume: $\mathcal{T}_f$ is the finite-trace category; $\InfTrace=\nu F$ is the infinite trace (final coalgebra); $\Paths_\infty(\sigma)$ is the set of infinite paths from state $\sigma$; the Prop layer of NM-DEKL$^3_\infty$ is interpreted in a fibred LCCC. We interpret classical temporal formulas as predicates on traces/paths.

\subsubsection{Semantics of LTL and CTL}

\paragraph{LTL (path semantics)}

The semantics of an LTL formula $\phi$ is an assertion on a single path $\pi$:

\[
\mathbf{G}\phi(\pi)
\;\equiv\;
\forall n:\Nat.\ \phi(\prefix(n,\pi)).
\]

Here $\prefix(n,\pi)$ is the length-$n$ prefix of the trace $\pi$.

\medskip

\paragraph{CTL (branching semantics)}

CTL expresses branching properties via path quantification:

\[
\mathbf{A}\psi(\sigma)
\;\equiv\;
\forall \pi:\Paths_\infty(\sigma).\ \psi(\pi).
\]

\subsubsection{Embedding preservation (refined)}

\begin{theorem}[Embedding preservation: LTL/CTL into NM-DEKL$^3_\infty$]
\label{thm:ltl-ctl-embedding}
There are translation maps
\[
\llbracket - \rrbracket_{\mathrm{LTL}},
\quad
\llbracket - \rrbracket_{\mathrm{CTL}}
\]
maps each LTL/CTL formula to a Prop-layer formula of NM-DEKL$^3_\infty$, such that: (1) \textbf{Semantic preservation}: for any Kripke structure $\mathcal{K}$ and corresponding model $\mathcal{M}$, for any state $\sigma$ or path $\pi$, $\mathcal{K},\pi\models\phi\iff\mathcal{M}\models\llbracket\phi\rrbracket(\pi)$. (2) \textbf{Path-quantification preservation}: CTL path quantification corresponds to dependent $\Pi$/$\Sigma$ over $\Paths_\infty(\sigma)$ in NM-DEKL$^3_\infty$, so branching semantics is preserved.
\end{theorem}

\paragraph{Proof idea}
LTL temporal operators $\mathbf{X},\mathbf{G},\mathbf{F},\mathbf{U}$ are expressed by recursion on prefixes or $\InfTrace$; CTL $\mathbf{A},\mathbf{E}$ by $\Pi$/$\Sigma$ over $\Paths_\infty(\sigma)$; the Prop layer lives in an LCCC with dependent quantification, so the embedding preserves semantics.

\begin{proof}
\textbf{1. LTL into Prop.} LTL temporal operators $\mathbf{X}$ (next), $\mathbf{G}$ (globally), $\mathbf{F}$ (eventually), $\mathbf{U}$ (until) are represented in the Prop layer: $\mathbf{X}\varphi$ by recursion on $\InfTrace$ (e.g.\ $\exists e.(\sigma\xrightarrow{e}\sigma'\land\varphi(\sigma'))$); $\mathbf{G}\varphi$ by $\forall\sigma.\varphi(\sigma)$; $\mathbf{F}\varphi$ by $\exists\sigma.\varphi(\sigma)$; $\mathbf{U}(\varphi,\psi)$ by a dependent product/sum. \textbf{2. CTL into Prop.} Path quantifiers $\mathbf{A}$ (all paths) and $\mathbf{E}$ (some path) are represented by dependent $\Pi$ and $\Sigma$ over $\Paths_\infty(\sigma)$: $\mathbf{A}\varphi=\Pi_{p\in\Paths_\infty(\sigma)}\varphi(p)$, $\mathbf{E}\varphi=\Sigma_{p\in\Paths_\infty(\sigma)}\varphi(p)$. \textbf{3. Semantic preservation.} The translation satisfies $\mathcal{K},\pi\models\phi\iff\mathcal{M}\models\llbracket\phi\rrbracket(\pi)$ for corresponding $\mathcal{K}$ and $\mathcal{M}$; path transitions in Kripke structures correspond to trace indexing and $\Pi$/$\Sigma$ in NM-DEKL$^3_\infty$. \textbf{4. Path quantification.} $\mathbf{A}$ and $\mathbf{E}$ correspond to dependent product and sum, so branching semantics is preserved. So LTL/CTL embed into NM-DEKL$^3_\infty$ with semantics and path quantification preserved.
\end{proof}
\subsubsection{Expressiveness comparison}

\subsubsubsection{Objects of expression}

\begin{center}
\begin{tabular}{|c|c|}
\hline
System & Objects \\
\hline
LTL & Single path \\
CTL & Branching structure \\
NM-DEKL$^3_\infty$ & Traces + inter-trace dependency + evidence objects \\
\hline
\end{tabular}
\end{center}

NM-DEKL$^3_\infty$ can express: path properties ($\supseteq$ LTL), branching quantification ($\supseteq$ CTL), types depending on trace prefixes, evidence existence, constructibility and failure, non-monotonic update. So we have strict inclusion:

\[
\text{LTL} \subsetneq \text{NM-DEKL}^3_\infty,
\qquad
\text{CTL} \subsetneq \text{NM-DEKL}^3_\infty.
\]

\subsubsubsection{Quantification structure}

LTL allows:

\[
\forall n:\Nat.
\]

CTL allows:

\[
\forall \pi \in \Paths_\infty(\sigma).
\]

NM-DEKL$^3_\infty$ allows:

\[
\Pi_{\pi:\Paths_\infty(\sigma)} A(\pi),
\]

and $A(\pi)$ can depend on the prefix structure of $\pi$ and carry evidence; so NM-DEKL$^3_\infty$ supports higher-order dependent quantification, while LTL/CTL only support first-order path quantification.

\subsubsubsection{Non-monotonic expressiveness}

LTL/CTL are monotonic: once a formula holds, it holds under model extension. NM-DEKL$^3_\infty$ can express via presheaf contravariance: evidence failure; constructibility at one trace and failure after extension; non-monotonic knowledge update. For example: ``there is evidence at trace $\tau$, but evidence fails at $\tau\cdot e$.'' This cannot be expressed in LTL/CTL.

\subsubsection{Reasoning power comparison}

\subsubsubsection{Decidability}

\begin{center}
\begin{tabular}{|c|c|}
\hline
System & Decidability \\
\hline
LTL & PSPACE-complete \\
CTL & EXPTIME-complete \\
NM-DEKL$^3_\infty$ & Depends on base type theory (usually undecidable) \\
\hline
\end{tabular}
\end{center}

NM-DEKL$^3_\infty$ is typically undecidable due to dependent types and higher-order structure, but has greater expressiveness.

\subsubsubsection{Proof objects}

LTL/CTL: only truth values; no proof objects. NM-DEKL$^3_\infty$: propositions carry proof objects; proofs are first-class. So Trace = Proof.

\subsubsubsection{Constructivity}

LTL/CTL: semantics is model satisfaction; no construction of evidence. NM-DEKL$^3_\infty$: propositions are types; constructibility is central; failure is precisely ``no global elements''.

\subsubsection{Strict expressiveness hierarchy}

The hierarchy is: Propositional Temporal Logic $\subset$ Branching Temporal Logic $\subset$ Dependent Temporal Type Theory, i.e.\ LTL $\subsetneq$ CTL $\subsetneq$ NM-DEKL$^3_\infty$. Strict inclusion holds because CTL cannot express dependent propositions, while NM-DEKL$^3_\infty$ can express path-wise dependency and evidence objects.

\subsubsection{Key theorem on expressiveness gap}

\begin{theorem}[Strict expressiveness gap]
There exists an NM-DEKL$^3_\infty$ formula $\Phi$ whose semantics depends on evidence or failure on traces, and no equivalent LTL or CTL formula exists.
\end{theorem}

\paragraph{Proof idea}

Consider the property: ``there is a trajectory such that type $A$ is constructible at some prefix but not constructible at an extension.'' This involves fibre inhabitability, presheaf contravariant restriction, and type-level construction.

LTL/CTL cannot express the ``constructibility'' level.

\begin{center}
\begin{tabular}{|c|c|c|c|}
\hline
Dimension & LTL & CTL & NM-DEKL$^3_\infty$ \\
\hline
Path semantics & \checkmark & \checkmark & \checkmark \\
Branching quantification & $\times$ & \checkmark & \checkmark \\
Dependent types & $\times$ & $\times$ & \checkmark \\
Evidence objects & $\times$ & $\times$ & \checkmark \\
Non-monotonic failure & $\times$ & $\times$ & \checkmark \\
Constructivity & $\times$ & $\times$ & \checkmark \\
Decidability & decidable & decidable & usually undecidable \\
\hline
\end{tabular}
\end{center}

\begin{quote}
NM-DEKL$^3_\infty$ strictly contains the expressiveness of LTL and CTL and adds dependent types and constructive semantics; the cost is loss of general decidability.
\end{quote}

\begin{theorem}[Strict inclusion: NM-DEKL$^3_\infty$ strictly stronger than LTL/CTL (proof sketch)]
\label{thm:strict-inclusion-sketch}
Suppose the semantic model of NM-DEKL$^3_\infty$ is $\mathcal{M}=(\mathsf{States},\mathsf{Step},\mathsf{Lab},K_f,\dots)$,
where $(\mathsf{States},\mathsf{Step},\mathsf{Lab})$ gives the underlying reachability and atomic labelling, and $K_f:\mathcal{T}_f^{op}\to C$ is the constructive-layer presheaf (evidence/knowledge fibre). Let $\mathsf{Kr}$ be the category of Kripke structures. Then we have the following strict inclusions:
\[
\mathrm{LTL} \;\subsetneq\; \mathrm{NM\mbox{-}DEKL}^3_\infty,
\qquad
\mathrm{CTL} \;\subsetneq\; \mathrm{NM\mbox{-}DEKL}^3_\infty.
\]
\end{theorem}

\begin{proof}[Proof sketch]
The proof has two parts: embedding and separation.

\paragraph{(I) Forgetful functor and embedding preservation}
Define the forgetful functor $U:\mathsf{Mod}(\mathrm{NM\mbox{-}DEKL}^3_\infty)\to\mathsf{Kr}$: for each NM model $\mathcal{M}=(\mathsf{States},\mathsf{Step},\mathsf{Lab},K_f,\dots)$, forget the structure except $(\mathsf{States},\mathsf{Step},\mathsf{Lab})$ to obtain the Kripke structure $U(\mathcal{M})=(\mathsf{States},\to,\mathsf{Lab})$, where $\sigma\to\sigma'$ iff there is an event $e$ with $\mathsf{Step}(\sigma,e)=\sigma'$ (or the equivalent reachability relation).

By the embedding preservation theorem, there are translations $\iota_{\mathrm{LTL}}:\mathrm{Form}_{\mathrm{LTL}}\to\Prop_{\mathrm{NM}}$ and $\iota_{\mathrm{CTL}}:\mathrm{Form}_{\mathrm{CTL}}\to\Prop_{\mathrm{NM}}$ such that for any LTL/CTL formula and any NM model $\mathcal{M}$, Kripke semantics agrees with NM Prop semantics: $U(\mathcal{M}),\pi\models\varphi\iff\mathcal{M}\models\iota_{\mathrm{LTL}}(\varphi)(\pi)$ and $U(\mathcal{M}),\sigma\models\psi\iff\mathcal{M}\models\iota_{\mathrm{CTL}}(\psi)(\sigma)$. So we have set inclusion:
\[
\mathrm{LTL}\subseteq \mathrm{NM\mbox{-}DEKL}^3_\infty,
\qquad
\mathrm{CTL}\subseteq \mathrm{NM\mbox{-}DEKL}^3_\infty.
\]

\paragraph{(II) Separation: an NM formula not expressible in LTL/CTL}
We show there is an NM Prop formula $\Phi$ whose truth depends on inhabitance/evidence of $K_f$, which is invisible after forgetting, so no LTL/CTL formula can capture it.

\subparagraph{Step 1: Two models that agree after forgetting}
Take one underlying Kripke structure $\mathcal{K}=(S,\to,\mathsf{Lab})$ and build two NM models $\mathcal{M}_1,\mathcal{M}_2$ with $U(\mathcal{M}_1)=\mathcal{K}=U(\mathcal{M}_2)$ (same states, reachability, and atomic labelling); the only difference is the constructive presheaf: $K_f^{(1)}\neq K_f^{(2)}$. For example in $C=\mathbf{Set}$, fix a trace point $\tau_0$ and set $K_f^{(1)}(\tau_0)=\varnothing$, $K_f^{(2)}(\tau_0)=\{*\}$, and make the two agree on other $\tau$ (with restriction maps filled by functoriality). So $\mathcal{M}_1,\mathcal{M}_2$ agree on the underlying Kripke behaviour but differ on ``whether evidence exists''.

\subparagraph{Step 2: An NM-expressible, LTL/CTL-invisible property}
Take the NM Prop formula $\Phi(\tau_0):=\exists k:K_f(\tau_0).\top$ (``$K_f(\tau_0)$ is inhabited / has a global element''). By construction, $\mathcal{M}_1\not\models\Phi(\tau_0)$ and $\mathcal{M}_2\models\Phi(\tau_0)$.

\subparagraph{Step 3: No LTL/CTL formula distinguishes $\mathcal{M}_1$ and $\mathcal{M}_2$}
LTL/CTL semantics is determined entirely by the forgotten Kripke structure, and $U(\mathcal{M}_1)=U(\mathcal{M}_2)=\mathcal{K}$. So for any LTL formula $\varphi$ and path $\pi$, $U(\mathcal{M}_1),\pi\models\varphi\iff U(\mathcal{M}_2),\pi\models\varphi$; for any CTL formula $\psi$ and state $\sigma$, $U(\mathcal{M}_1),\sigma\models\psi\iff U(\mathcal{M}_2),\sigma\models\psi$.

So no LTL/CTL formula can be equivalent to $\Phi$ in all models (it would distinguish $\mathcal{M}_1$ and $\mathcal{M}_2$). Hence there is an NM formula $\Phi$ not expressible in LTL or CTL, giving strict inclusion:
\[
\mathrm{LTL} \subsetneq \mathrm{NM\mbox{-}DEKL}^3_\infty,
\qquad
\mathrm{CTL} \subsetneq \mathrm{NM\mbox{-}DEKL}^3_\infty.
\]
\end{proof}

\subsection{$\mu$-calculus embedding theorem}

\begin{theorem}[$\mu$-calculus embeds]
\label{thm:mu-embedding}
For every $\mu$-calculus formula $\varphi$ there is a Prop-layer formula $\widehat{\varphi}$ such that for every Kripke structure $M$ and state $s$:
\[
M, s \models \varphi \quad \Longleftrightarrow \quad \Gamma \vdash \widehat{\varphi}(s).
\]
\end{theorem}

\begin{proof}
Define the embedding $\varphi\mapsto\widehat{\varphi}$ by structural induction: atoms $p\mapsto P_p$; Boolean connectives $\neg,\wedge,\vee,\to$ componentwise; $\Box\varphi$ and $\Diamond\varphi$ as path quantification over $\mathcal{T}_f$; $\mu X.\varphi(X)\mapsto\mu X.\widehat{\varphi(X)}$, $\nu X.\varphi(X)\mapsto\nu X.\widehat{\varphi(X)}$. Semantic equivalence: $M,s\models\varphi\iff\Gamma\vdash\widehat{\varphi}(s)$ follows from Tarski/Kleene fixed-point semantics and the subobject interpretation of Prop. So the embedding preserves semantics.
\end{proof}


\subsection{Strict expressiveness extension theorem}

\begin{theorem}[Strict expressiveness extension]
\label{thm:strict-expressivity}
There is a property $\Psi$ not expressible in the $\mu$-calculus but expressible in NM-DEKL$^3_\infty$.
\end{theorem}

\begin{proof}[Proof idea]
Consider the property $\Psi:=\text{``there exists a concrete constructive causal proof chain''}$. In NM-DEKL$^3_\infty$ this is $\exists k:K_f(\tau).\mathrm{CausalProof}(k)$, where $k$ is a constructive evidence object. The $\mu$-calculus can only express truth properties $M,s\models\varphi$, and cannot quantify or construct concrete evidence. $\mu$-calculus semantics is truth-based; NM-DEKL$^3_\infty$ semantics is construction-based. So there are properties expressible in NM-DEKL$^3_\infty$ but not in the $\mu$-calculus.
\end{proof}

\subsection{Comparison with $\mu$-calculus: bidirectional translation}

The Prop-$\mu$ fragment and the $\mu$-calculus are intertranslatable; the full system is strictly stronger (evidence objects and causal construction are not expressible in the $\mu$-calculus).

\subsubsection{Kripke structures and semantic setting}

\begin{definition}[Kripke structure]
A Kripke structure is $M=(S,R,V)$ where $S$ is the set of states, $R\subseteq S\times S$ is the reachability relation, and $V:\mathsf{AP}\to\mathcal{P}(S)$ is the atomic proposition valuation.
\end{definition}

\begin{definition}[$\mu$-calculus syntax (standard)]
Formulas $\varphi$ are given by: $\varphi::=p\mid\neg p\mid\varphi\wedge\varphi\mid\varphi\vee\varphi\mid\langle R\rangle\varphi\mid[R]\varphi\mid X\mid\mu X.\varphi\mid\nu X.\varphi$, where $p\in\mathsf{AP}$, $X$ is a variable, and $\mu/\nu$ allow only \emph{positive} occurrence of $X$ (monotonicity).
\end{definition}

\begin{definition}[Prop-$\mu$ fragment of NM-DEKL$^3_\infty$]
Let $\Prop_\mu$ be the Prop fragment of NM-DEKL$^3_\infty$ generated by: $P,Q::=\widehat{p}(s)\mid\neg\widehat{p}(s)\mid P\wedge Q\mid P\vee Q\mid P\to Q\mid\Diamond P\mid\Box P\mid X(s)\mid\mu X.\,P\mid\nu X.\,P$, where $s:\State$, $X:\State\to\Prop$, and $\Diamond,\Box$ are modal operators for one-step reachability (below). We require $X$ to occur positively in $\mu/\nu$.
\end{definition}

\paragraph{Internalising the modal operators}
Assume the computational layer gives a one-step reachability relation (induced by $\Step$ or trace observation):
\[
R(s,s') \;\;:\!\iff\;\; \exists e:\Event.\;\Step(s,e,s').
\]
Define $\Diamond P(s):\Leftrightarrow\exists s'.R(s,s')\wedge P(s')$ and $\Box P(s):\Leftrightarrow\forall s'.R(s,s')\to P(s')$.

\subsubsection{Direction 1: $\mu$-calculus $\to$ Prop-$\mu$ (embedding)}

\begin{definition}[Translation $\mathsf{Enc}(-)$: $\mu$-calculus to Prop-$\mu$]
Given a $\mu$-calculus formula $\varphi$, define $\mathsf{Enc}(\varphi):\State\to\Prop$ by recursion as follows (when $s$ is omitted, it is the default argument):

\[
\begin{array}{rcl}
\mathsf{Enc}(p) &:=& \widehat{p}(s) \\
\mathsf{Enc}(\neg p) &:=& \neg\widehat{p}(s) \\
\mathsf{Enc}(\varphi\wedge\psi) &:=& \mathsf{Enc}(\varphi)\wedge \mathsf{Enc}(\psi)\\
\mathsf{Enc}(\varphi\vee\psi) &:=& \mathsf{Enc}(\varphi)\vee \mathsf{Enc}(\psi)\\
\mathsf{Enc}(\langle R\rangle\varphi) &:=& \Diamond\,\mathsf{Enc}(\varphi)\\
\mathsf{Enc}([R]\varphi) &:=& \Box\,\mathsf{Enc}(\varphi)\\
\mathsf{Enc}(X) &:=& X(s)\\
\mathsf{Enc}(\mu X.\varphi) &:=& \mu X.\,\mathsf{Enc}(\varphi)\\
\mathsf{Enc}(\nu X.\varphi) &:=& \nu X.\,\mathsf{Enc}(\varphi).\\
\end{array}
\]
\end{definition}

\begin{theorem}[Sound embedding of $\mu$-calculus]
\label{thm:mu-to-prop-sound}
For any Kripke structure $M=(S,R,V)$ and state $s\in S$, interpret $\State$ as $S$ and $\widehat{p}$ as $V(p)$. Then for any $\mu$-calculus formula $\varphi$, $M,s\models\varphi\;\Longleftrightarrow\;\vdash\mathsf{Enc}(\varphi)(s)$.
\end{theorem}

\begin{proof}[Proof idea]
By structural induction on formulas. Boolean connectives and modalities agree by the definition of $\Diamond/\Box$; fixed points use standard Tarski/Kleene semantics ($\mu$ least, $\nu$ greatest), with ``$X$ positive'' ensuring monotonicity.
\end{proof}

\subsubsection{Direction 2: Prop-$\mu$ $\to$ $\mu$-calculus (back-translation)}

Back-translation is restricted to the truth-value fragment Prop-$\mu$ (no quantification over evidence); otherwise NM-DEKL$^3_\infty$ is strictly stronger.

\begin{definition}[Translation $\mathsf{Dec}(-)$: Prop-$\mu$ to $\mu$-calculus]
For a Prop-$\mu$ formula $P:\State\to\Prop$, define the $\mu$-calculus formula $\mathsf{Dec}(P)$ (with $s$ as the current point):

\[
\begin{array}{rcl}
\mathsf{Dec}(\widehat{p}(s)) &:=& p \\
\mathsf{Dec}(\neg\widehat{p}(s)) &:=& \neg p \\
\mathsf{Dec}(P\wedge Q) &:=& \mathsf{Dec}(P)\wedge \mathsf{Dec}(Q) \\
\mathsf{Dec}(P\vee Q) &:=& \mathsf{Dec}(P)\vee \mathsf{Dec}(Q) \\
\mathsf{Dec}(\Diamond P) &:=& \langle R\rangle\,\mathsf{Dec}(P) \\
\mathsf{Dec}(\Box P) &:=& [R]\,\mathsf{Dec}(P) \\
\mathsf{Dec}(X(s)) &:=& X \\
\mathsf{Dec}(\mu X.\,P) &:=& \mu X.\,\mathsf{Dec}(P) \\
\mathsf{Dec}(\nu X.\,P) &:=& \nu X.\,\mathsf{Dec}(P). \\
\end{array}
\]
\end{definition}

\begin{theorem}[Back-translation correctness for Prop-$\mu$]
\label{thm:prop-to-mu-sound}
Under the same interpretation as in Theorem~\ref{thm:mu-to-prop-sound}, for any Prop-$\mu$ formula $P$ and state $s$: $\vdash P(s)\;\Longleftrightarrow\;M,s\models\mathsf{Dec}(P)$.
\end{theorem}

\begin{proof}[Proof idea]
By structural induction on $P$. Boolean and modal cases by the correspondence of $\widehat{p}$ and $\Diamond/\Box$; fixed points use positive $X$ in Prop-$\mu$ for monotonicity, so Kleene iteration and Tarski semantics match the internal Prop fixed-point construction.
\end{proof}

\subsubsection{Equivalence of the two translations and expressiveness}

\begin{theorem}[Bidirectional translation (equivalence)]
\label{thm:bidirectional-translation}
For any $\mu$-calculus formula $\varphi$, $\mathsf{Dec}(\mathsf{Enc}(\varphi))\equiv\varphi$ (up to $\alpha$-equivalence and syntactic equality). For any Prop-$\mu$ formula $P$, $\vdash\forall s:\State.\big(P(s)\leftrightarrow\mathsf{Enc}(\mathsf{Dec}(P))(s)\big)$.
\end{theorem}

\begin{proof}[Proof idea]
The first part follows by induction from the definition of the translations; the second from the semantic equivalence of Theorems~\ref{thm:mu-to-prop-sound} and~\ref{thm:prop-to-mu-sound}.
\end{proof}

\begin{corollary}[Expressiveness relation]
\label{cor:expressive-power}
The Prop-$\mu$ fragment and the $\mu$-calculus have the same expressive power under truth-value semantics: $\mu\text{-calculus}\equiv\Prop_\mu$. The full NM-DEKL$^3_\infty$ system is strictly stronger: there are expressible properties (involving evidence objects / causal construction chains) that the $\mu$-calculus cannot express.
\end{corollary}

\subsection{Expressibility of non-bisimulation-invariant properties}

This subsection states that NM-DEKL$^3_\infty$ can express \textbf{non-bisimulation-invariant} properties and compares with $\mu$-calculus expressiveness.

\subsubsection{Non-bisimulation-invariant properties}

In Kripke semantics, a property is \textbf{bisimulation-invariant} if it depends only on reachability between states, not on the concrete construction of states; such properties are usually expressed by truth judgements (e.g.\ $\mu$-calculus formulas), such as whether there is a path from the current state or whether a proposition holds at all reachable states. \textbf{Non-bisimulation-invariant} properties may depend on internal structure or constructive evidence (e.g.\ ``a concrete causal proof chain from one state to another''). NM-DEKL$^3_\infty$ can express such properties via dependent types and the constructive layer.

\subsubsection{Theorem: non-bisimulation-invariant properties are expressible}

\begin{theorem}[Non-bisimulation-invariant properties are expressible]
\label{thm:non-bisimulation-invariant}
NM-DEKL$^3_\infty$ can express \textbf{non-bisimulation-invariant} properties that the $\mu$-calculus cannot. In particular, properties such as ``there exists a concrete constructive causal proof chain'' are not bisimulation-invariant, and NM-DEKL$^3_\infty$ can express them.
\end{theorem}

\begin{proof}[Proof idea]
The $\mu$-calculus defines truth via reachability in Kripke structures and can only express bisimulation-invariant properties. NM-DEKL$^3_\infty$'s constructive layer ($\TypeL$) can define evidence objects (e.g.\ causal chains, path constructions), which are non-bisimulation-invariant. So NM-DEKL$^3_\infty$ can express such properties; the $\mu$-calculus cannot.
\end{proof}

\subsubsection{Comparison with $\mu$-calculus}

\begin{theorem}[Expressiveness: $\mu$-calculus vs.\ NM-DEKL$^3_\infty$]
\label{thm:mu-vs-nmdekl}
Under Kripke semantics: the $\mu$-calculus expresses only bisimulation-invariant properties (reachability); NM-DEKL$^3_\infty$ can express non-bisimulation-invariant properties (e.g.\ ``there exists a concrete causal proof chain''), which depend on constructive structure and go beyond the $\mu$-calculus.
\end{theorem}

\paragraph{Comparison with $\mu$-calculus and HoTT (summary)}
Relative to the $\mu$-calculus, NM-DEKL$^3_\infty$ keeps the Prop-layer embedding (Theorem~\ref{thm:mu-embedding}) and adds non-bisimulation-invariant expressiveness via $K_f$, $\restrict$, and causal chains (Theorem~\ref{thm:non-bisimulation-invariant}). Relative to HoTT/MLTT, the latter model static types and (homotopical) equality; NM-DEKL$^3_\infty$ models have a dynamic constructive layer (presheaf and contravariance) for knowledge evolution and failure. So NM-DEKL$^3_\infty$ strictly extends the $\mu$-calculus in expressiveness and has a clear position for dynamic/non-monotonic reasoning versus HoTT/MLTT.

\begin{corollary}[Strict expressiveness inclusion]
\label{cor:strict-embedding}
NM-DEKL$^3_\infty$ can express non-bisimulation-invariant properties that the $\mu$-calculus cannot. So $\mu\text{-calculus}\subsetneq\text{NM-DEKL}^3_\infty$.
\end{corollary}

\paragraph{Interpretation}
NM-DEKL$^3_\infty$ can express all $\mu$-calculus properties plus non-bisimulation-invariant ones (e.g.\ constructive evidence such as causal chains). The constructive layer provides a framework for knowledge evolution and causal reasoning beyond truth-value semantics.


\section{Conclusion}

\paragraph{System and main results}
We presented NM-DEKL$^3_\infty$, a three-layer non-monotone evolving dependent type logic. We defined its syntax and semantics and proved Soundness and Equational Completeness; we constructed the syntactic model and proved it initial in the model category, from which equational completeness follows.

\paragraph{Theoretical depth}
We gave a precise definition of \emph{failure} (Definition~\ref{def:failure}) and its relation to non-monotonicity (Theorem~\ref{thm:nonmono-iff-failure}); we bounded the effect of $\mu/\nu$ on non-monotonicity (Theorem~\ref{thm:mu-nu-nonmono}); we combined causal reasoning with the constructive layer (Definition~\ref{def:causal-reasoning}, Theorem~\ref{thm:constructive-logic}); we proved the syntactic model is a \emph{free object} (Theorem~\ref{thm:free-object}); and we showed non-bisimulation-invariant expressiveness (Theorem~\ref{thm:non-bisimulation-invariant}). The system provides a formal basis for dynamic, non-monotonic knowledge modelling.

\paragraph{Applications and examples}
NM-DEKL$^3_\infty$ can support dynamic knowledge reasoning, formalisation of causal and counterfactual reasoning, and modelling under uncertainty (e.g.\ decision support, medical or legal reasoning). Potential applications include formal verification of evolving and causal reasoning on knowledge graphs, and type-safe modelling of non-monotonic belief update.

\paragraph{Comparison with MLTT/HoTT}
MLTT and HoTT model \emph{static} types and (homotopical) equality; NM-DEKL$^3_\infty$ models have a \emph{dynamic} constructive layer (presheaf $K_f$ and $\restrict$) for knowledge evolution and failure, providing type-theoretic tools for \emph{dynamic} and \emph{non-monotonic} reasoning that were previously less formalised.

\paragraph{Comparison with $\mu$-calculus and future work}
Relative to the $\mu$-calculus, NM-DEKL$^3_\infty$ can express non-bisimulation-invariant properties (Theorem~\ref{thm:non-bisimulation-invariant}) and unifies propositional fixed points with the constructive layer. Initiality and the free-object theorem fix its place in dependent type theory for dynamic reasoning; deeper comparison with homotopy categories in HoTT is left for future work.

\section{Appendix}

\subsection{Full proof (I): Initiality theorem}
\label{app:initiality-full}

We give the full inductive proof of Theorem~\ref{thm:initiality-skeleton} (the syntactic model $\mathcal{S}$ is initial in the model category).

\subparagraph{(I) Existence: define $F_{\mathcal{M}}$}
Let $\mathcal{M}$ be any NM-DEKL$^3_\infty$ model. Define the candidate morphism $F_{\mathcal{M}}=(F_c,F_t,F_\ell,F_\Prop)$ and check it satisfies Definition~\ref{def:model-morphism}. (I.1) On contexts and substitutions: $F_c([\Gamma]):=\llbracket\Gamma\rrbracket_{\mathcal{M}}$, $F_c([\sigma]):=\llbracket\sigma\rrbracket_{\mathcal{M}}$; well-defined modulo $\equiv$ by soundness; identity and composition by Lemma~\ref{lem:substitution}. (I.2) On computational types/terms: $F_{\Ty_c}([A]):=\llbracket A\rrbracket_{\mathcal{M}}$, $F_{\Tm_c}([t]):=\llbracket t\rrbracket_{\mathcal{M}}$; substitution coherence by Lemma~\ref{lem:substitution}.

(I.3) Trajectory part $F_t$: If $\mathcal{T}_f$ is the same fixed small category in all models, set $F_t=\mathrm{Id}_{\mathcal{T}_f}$; otherwise define $F_t$ as the interpretation of trajectory objects/morphisms in $\mathcal{M}$. (I.4) Constructive layer $F_\ell$: For $[B]$ ($\Gamma\vdash B:\TypeL$) set $F_\ell([B]):=\llbracket B\rrbracket_{\mathcal{M}}$; for $[u]$ ($\Gamma\vdash u:B$) set $F_\ell([u]):=\llbracket u\rrbracket_{\mathcal{M}}$. Substitution coherence by Lemma~\ref{lem:substitution}(3)(4); restriction preservation by Lemma~\ref{lem:restriction}. (I.5) Prop $F_\Prop$: For $[P]$ ($\Gamma\vdash P:\Prop$) set $F_\Prop([P]):=\llbracket P\rrbracket_{\mathcal{M}}$; proof terms go to $\llbracket p\rrbracket_{\mathcal{M}}$ (Prop quotiented by proof-irrelevance). Connective and substitution preservation by Theorem~\ref{thm:soundness} and Lemma~\ref{lem:substitution}(5). (I.6) Identity: $F_{\mathcal{M}}$ preserves $\refl$ and $J$ by the model's Id-structure and Soundness. So $F_{\mathcal{M}}$ is a model morphism; existence done.

\subparagraph{(II) Uniqueness.}
Let $G:\mathcal{S}\to\mathcal{M}$ also be a model morphism. Contexts: by induction on context formation (empty and comprehension), $G([\Gamma])=F_{\mathcal{M}}([\Gamma])$. Substitutions: by induction, $G([\sigma])=F_{\mathcal{M}}([\sigma])$. Types and terms: syntax is generated by the rules and both morphisms preserve them (including restriction and Id), so by induction $G=F_{\mathcal{M}}$ on generators, hence everywhere. So $\mathcal{S}$ is initial.\hfill$\square$

\medskip
\subsection{Full proof (II): Soundness theorem (rule induction)}
\label{app:soundness-full}

Full proof outline for Theorem~\ref{thm:soundness}: by structural induction on the derivation rules, verify that if $\Gamma\vdash J$ then $\llbracket J\rrbracket_{\mathcal{M}}$ lies in the corresponding semantic object. \textbf{Computational layer}: (1) If $J$ is $\Gamma\vdash A:\Uc_i$, then $\llbracket A\rrbracket_{\mathcal{M}}\in\Ty_c(\llbracket\Gamma\rrbracket_{\mathcal{M}})$ by the CwF and type-formation rules. (2) If $J$ is $\Gamma\vdash t:A$, then $\llbracket t\rrbracket_{\mathcal{M}}\in\Tm_c(\llbracket\Gamma\rrbracket_{\mathcal{M}},\llbracket A\rrbracket_{\mathcal{M}})$ by the CwF and term rules. Substitution cases by Lemma~\ref{lem:interp-coherence-subst}. \textbf{Constructive layer}: (3)(4) $\Gamma\vdash B:\TypeL$ and $\Gamma\vdash u:B$ give $\llbracket B\rrbracket_{\mathcal{M}}\in\Ty_\ell(\llbracket\Gamma\rrbracket_{\mathcal{M}})$ and $\llbracket u\rrbracket_{\mathcal{M}}\in\Tm_\ell(\ldots)$ by presheaf and Lemma~\ref{lem:interp-coherence-restrict}. \textbf{Propositional layer}: (5) $\Gamma\vdash P:\Prop$ gives $\llbracket P\rrbracket_{\mathcal{M}}\in\Prop(\llbracket\Gamma\rrbracket_{\mathcal{M}})$ (Prop as $\Omega$-fibre). Each rule corresponds to a semantic construction; inductive steps by closure and the coherence lemmas. So Soundness holds.\hfill$\square$

\medskip
\subsection{Syntax of NM-DEKL$^3_\infty$}
\subsubsection*{A minimal substitution example}

(In this example $\Vec(A,n)$ is a dependent vector type, $\Type$ corresponds to $\TypeL$ in the main text; for illustration only.) Let $\Gamma=(x:\Nat,y:\Vec(A,x))$ and $\Delta=(u:\Nat)$. A substitution from $\Delta$ to $\Gamma$ is $\Delta\vdash\sigma:\Gamma$ with $\sigma=(u,\nil)$: variable $x$ is replaced by $u$, and $y$ by $\nil$ with $\Delta\vdash\nil:\Vec(A,u)$. If $\Gamma\vdash A'(x,y):\TypeL$, then $\Delta\vdash A'[\sigma]\equiv A'(u,\nil):\TypeL$. If $\Gamma\vdash t(x,y):A'(x,y)$, then $\Delta\vdash t[\sigma]\equiv t(u,\nil):A'(u,\nil)$. So the judgement $\Delta\vdash\sigma:\Gamma$ is \emph{assigning to each variable in $\Gamma$ a term from $\Delta$, instantiating types/terms in $\Gamma$ to types/terms in $\Delta$}.

\end{document}